\newif\ifdraft \draftfalse
\newif\iffull \fullfalse
\newif\ifcomment \commenttrue

\documentclass[11pt]{article}

\ifcomment\else \usepackage{kpfonts} \fi
\usepackage{algorithm}
\usepackage[noend]{algorithmic}
\usepackage{amsmath, amssymb, amsthm}
\usepackage[numbers]{natbib}
\usepackage{mathtools}
\usepackage{verbatim}
\usepackage{color}
\usepackage{xcolor}
\usepackage{fullpage}
\usepackage{multicol}

\definecolor{DarkGreen}{rgb}{0.1,0.5,0.1}
\definecolor{DarkRed}{rgb}{0.5,0.1,0.1}
\definecolor{DarkBlue}{rgb}{0.1,0.1,0.5}
\usepackage[]{hyperref}
\hypersetup{
    unicode=false,          
    pdftoolbar=true,        
    pdfmenubar=true,        
    pdffitwindow=false,      
    pdftitle={},    
    pdfauthor={}
    pdfsubject={},   
    pdfnewwindow=true,      
    pdfkeywords={keywords}, 
    colorlinks=true,       
    linkcolor=DarkRed,          
    citecolor=DarkGreen,        
    filecolor=DarkRed,      
    urlcolor=DarkBlue,          
}
\usepackage{xspace}
\usepackage{cleveref}
\usepackage{thmtools, thm-restate}

\newcommand{\katrina}[1]{\ifcomment{}\else\textcolor{brown}{[Katrina: {#1}]}\fi}
\newcommand{\sw}[1]{\ifcomment{}\else\textcolor{cyan}{[Steven: {#1}]}\fi}

\newcommand{\rc}[1]{\ifcomment{}\else\textcolor{blue}{[Rachel: #1]}\fi}
\newcommand{\todo}[1]{\ifcomment{}\else\textcolor{magenta}{[To Do: #1]}\fi}

\newcommand\NN{\mathbb{N}}
\newcommand\RR{\mathbb{R}}
\newcommand\cA{\mathcal{A}}
\newcommand\cO{\mathcal{O}}
\newcommand\cH{\mathcal{H}}
\newcommand\cL{\mathcal{L}}
\newcommand\cC{\mathcal{C}}

\newcommand\cE{\mathcal{E}}

\newcommand\cJ{\mathcal{J}}

\newcommand\cM{\mathcal{M}}

\newcommand\cR{\mathcal{R}}
\newcommand\cX{\mathcal{X}}
\newcommand\cY{\mathcal{Y}}
\newcommand\cXl{\mathcal{X}_L}
\newcommand\cXL{\cXl}
\newcommand\cD{\mathcal{D}}

\newcommand\E{\Expectation}

\newcommand\R{\mathbb{R}}

\DeclareMathOperator{\poly}{poly}

\renewcommand{\tilde}{\widetilde}

\DeclareMathOperator*{\Expectation}{\mathbb{E}}
\newcommand{\Ex}[2]{\Expectation_{#1}\left[#2\right]}

\newcommand{\vc}{{\sf VCDIM}}

\newcommand{\eps}{\varepsilon}
\def\epsilon{\varepsilon}

\DeclareMathOperator{\Lap}{Lap}
\DeclareMathOperator{\OPT}{OPT}
\DeclareMathOperator{\Supp}{Supp}

 \theoremstyle{definition}
 \newtheorem{thm}{Theorem}[section]
 \newtheorem{definition}[thm]{Definition}
  \newtheorem{lemma}[thm]{Lemma}
   \newtheorem{corollary}[thm]{Corollary}
    \newtheorem{theorem}[thm]{Theorem}
  
    \newtheorem{remark}[thm]{Remark}

\usepackage{kpfonts}

\title{Adaptive Learning with Robust Generalization Guarantees}

\author{Rachel Cummings\thanks{Dept.\ of Computing and Mathematical Sciences, California Institute of Technology. {\tt rachelc@caltech.edu}. Supported in part by NSF grant 1254169, US-Israel Binational Science Foundation grant 2012348, and a Simons Graduate Fellowship.} \and
Katrina Ligett\thanks{Dept.\ of Computing and Mathematical Sciences, California Institute of Technology and Benin School of Computer Science and Engineering, Hebrew University of Jerusalem. {\tt katrina@caltech.edu}
Supported in part by
NSF grants 1254169 and 1518941,
US-Israel Binational Science Foundation Grant 2012348,
the Charles Lee Powell Foundation,
a Google Faculty Research Award,
an Okawa Foundation Research Grant,
a subcontract through the DARPA Brandeis project,
a grant from the HUJI Cyber Security Research Center, and
a startup grant from Hebrew University's School of Computer Science.
Part of this work was completed during a stay at the Simons Institute for the Theory of Computing at Berkeley.} \and
Kobbi Nissim\thanks{Dept.\ of Computer Science, Ben-Gurion University and Center for Research in Computation and Society, Harvard University. {\tt kobbi@seas.harvard.edu}. Supported by grants from the Sloan Foundation, a Simons Investigator grant to Salil Vadhan, and NSF grant CNS-1237235.} \and
Aaron Roth\thanks{Dept.\ of Computer and Information Sciences, University of Pennsylvania. {\tt aaroth@cis.upenn.edu}. Supported in part by an NSF CAREER award, NSF grant CNS-1513694, a subcontract through the DARPA Brandeis project, and a grant from the Sloan Foundation.} \and
Zhiwei Steven Wu\thanks{Dept.\ of Computer and Information Sciences, University of Pennsylvania. {\tt wuzhiwei@cis.upenn.edu}}}

\begin{document}

\maketitle

\begin{abstract}
  The traditional notion of \emph{generalization}---i.e., learning a
  hypothesis whose empirical error is close to its true error---is
  surprisingly brittle. As has recently been noted~\citep{stoc15},
  even if several algorithms have this guarantee in isolation, the
  guarantee need not hold if the algorithms are composed
  adaptively. In this paper, we study three notions of
  generalization---increasing in strength---that are \emph{robust} to
  postprocessing and amenable to adaptive composition, and examine the
  relationships between them.

  We call the weakest such notion \emph{Robust Generalization}. A
  second, intermediate, notion is the stability guarantee known as
  \emph{differential privacy}. The strongest guarantee we consider we
  call \emph{Perfect Generalization}. We prove that every hypothesis
  class that is PAC learnable is also PAC learnable in a robustly
  generalizing fashion, with almost the same sample complexity. It was
  previously known that differentially private algorithms satisfy
  robust generalization. In this paper, we show that robust
  generalization is a strictly weaker concept, and that there is a
  learning task that can be carried out subject to robust
  generalization guarantees, yet cannot be carried out subject to
  differential privacy. We also show that perfect generalization is a
  strictly stronger guarantee than differential privacy, but that,
  nevertheless, many learning tasks can be carried out subject to the
  guarantees of perfect generalization.
\end{abstract}

\section{Introduction}

Generalization, informally, is the ability of a learner to reflect not just its training data, but properties of the underlying distribution from which the data are drawn. When paired with empirical risk minimization, it is one of the fundamental tools of learning. Typically, we say that a learning algorithm {\em generalizes} if, given access to some training set drawn i.i.d.~from an underlying data distribution, it returns a hypothesis whose empirical error (on the training data) is close to its true error (on the underlying distribution).

This is, however, a surprisingly brittle notion---even if the output of a learning algorithm generalizes, one may be able to extract additional hypotheses by performing further computations on the output hypothesis---i.e., by postprocessing---that do not themselves generalize. As an example, notice that the standard notion of generalization does not prevent a learner from encoding the entire training set in the hypothesis that it outputs, which in turn allows a data analyst to generate a hypothesis that over-fits to an arbitrary degree.
In this sense, {\em traditional generalization is not robust to misinterpretation by subsequent analyses (postprocessing)} (either malicious or naive).

Misinterpretation of learning results is only one face of the threat---the problem is much more alarming. Suppose the output of a (generalizing) learning algorithm influences, directly or indirectly, the choice of future learning tasks.
For example, suppose a scientist chooses a scientific hypothesis to explore on some data, on the basis of previously (generalizingly!) learned correlations in that data set. Or suppose a data scientist repeatedly iterates a model selection procedure while validating it on the same holdout set, attempting to optimize his empirical error. These approaches are very natural, but also can lead to false discovery in the first case, and disastrous overfitting to the holdout set in the second~\citep{DFHPRR15science},  because {\em traditional generalization is not robust to adaptive composition}.

In this paper, we study two refined notions of generalization---\emph{robust generalization} and \emph{perfect generalization}, each of which is preserved under post-processing (we discuss their adaptive composition guarantees more below). Viewed in relation to these two notions, \emph{differential privacy} can also be cast as a third, intermediate generalization guarantee. It was previously known that differentially private algorithms were also robustly generalizing \citep{stoc15,BNSSSU15}. As we show in this paper, however, differential privacy is a strictly stronger guarantee---there are proper learning problems that can be solved subject to robust generalization that cannot be solved subject to differential privacy (or with any other method previously known to guarantee robust generalization). Moreover, we show that every PAC learnable class (even over infinite data domains) is learnable subject to robust generalization, with almost no asymptotic blowup in sample complexity (a comparable statement is not known for differentially private algorithms, and is known to be false for algorithms satisfying \emph{pure} differential privacy). We also show that, in a sense, differential privacy is a strictly weaker guarantee than perfect generalization. We provide a number of generic techniques for learning under these notions of generalization and prove useful properties for each. As we will discuss, perfect generalization also can be interpreted as a {\em privacy guarantee}, and thus may also be of interest to the privacy community.

\subsection{Our Results}
Informally, we say that a learning algorithm has a guarantee of \emph{robust generalization} if it is not only guaranteed to output a hypothesis whose empirical error is close to the true error (and near optimal), but if no adversary taking the output hypothesis as input can find another hypothesis whose empirical error differs substantially from its true error. (In particular, robustly generalizing algorithms are inherently robust to post-processing, and hence can be used to generate other test statistics in arbitrary ways without worry of overfitting). We say that a learning algorithm has the stronger guarantee of \emph{perfect generalization} if its output reveals almost nothing about the training data that could not have been learned via only direct oracle access to the underlying data distribution.

It was previously known \citep{stoc15,nips15,BNSSSU15} that both
\emph{differential privacy} and \emph{bounded description length
  outputs} are sufficient conditions to guarantee that a learning
algorithm satisfies robust generalization. However, prior to this
work, it was possible that differential privacy was \emph{equivalent}
to robust generalization in the sense that any learning problem that
could be solved subject to the guarantees of robust generalization
could also be solved via a differentially private
algorithm.\footnote{More precisely, it was known that algorithms with
  bounded description length could give robust generalization
  guarantees for the computation of high sensitivity statistics that
  could not be achieved via differential privacy
  \citep{nips15}. However, for low-sensitivity statistics (like the
  empirical error of a classifier, and hence for the problem of
  learning), there was no known separation.} Indeed, this was one of
the open questions stated in \cite{nips15}. We resolve this question
(Section~\ref{sec:thre}) by showing a simple proper learning task
(learning threshold functions over the real line) that can be solved
with guarantees of robust generalization (indeed, with the optimal
sample complexity) but that cannot be non-trivially properly learned
by any differentially private algorithm (or any algorithm with bounded
description length outputs). We do so (Theorem~\ref{thm:compressRG})
by showing that generalization guarantees that follow from
\emph{compression schemes} \citep{LW86} carry over to give guarantees
of robust generalization (thus giving a third technique, beyond
differential privacy and description length arguments, for
establishing robust generalization). In addition to threshold
learning, important learning procedures like SVMs have optimal
compression schemes, and so satisfy robust generalization without
modification. We also show (Theorem~\ref{thm:comcomp}) that
compression schemes satisfy an adaptive composition theorem, and so
can be used for adaptive data analysis while guaranteeing robust
generalization. Note that, somewhat subtly, robustly generalizing
algorithms derived by other means need not necessarily maintain their
robust generalization guarantees under adaptive composition (a
sequence of computations in which later computations have access not
only to the training data, but also to the outputs of previous
computations).  Using the fact that boosting implies the existence of
a near optimal variable-length compression scheme for every VC-class
(see \cite{DMY16}), we show (Theorem~\ref{thm:all}) that any PAC
learnable hypothesis class (even over an infinite domain) is also
learnable with robust generalization, with at most a logarithmic
blowup in sample complexity. (In fact, merely \emph{subsampling} gives
a simple ``approximate compression scheme'' for any VC-class, but one
that would imply a quadratically suboptimal sample complexity bound.
In contrast, we show that almost no loss in sample complexity --on top
of the sample complexity needed for outputting an accurate
hypothesis-- is necessary in order to get the guarantees of robust
generalization.) \sw{revamped this part}

We then show (Theorem \ref{thm.pglearning}) that perfectly
generalizing algorithms can be compiled into differentially private
algorithms (in a black box way) with little loss in their parameters,
and that (Theorem \ref{thm.dptosg}) differentially private algorithms
are perfectly generalizing, but with a loss of a factor of $\sqrt{n}$
in the generalization parameter. Moreover, we show (Theorem
\ref{thm.dptosgtight}) that this $\sqrt{n}$ loss is necessary. Because
differentially private algorithms satisfy an adaptive composition
theorem, this gives a method for designing perfectly generalizing
algorithms that are robust to arbitrary adaptive composition. Despite
this $\sqrt{n}$ blowup in the generalization parameter, we show
(Section~\ref{sec:generic_pg_learner}) that any \emph{finite}
hypothesis class can be PAC learned subject to perfect generalization.

\subsection{Related work}
Classically, machine learning has been concerned only with the basic
generalization guarantee that the empirical error of the learned
hypothesis be close to the true error. There are three main approaches
to proving standard generalization guarantees of this sort. The first
is by bounding various notions of complexity of the range of the
algorithm---most notably, the VC-dimension (see, e.g., \cite{KV94} for
a textbook introduction). These guarantees are \emph{not} robust to
post-processing or adaptive composition. The second follows from an
important line of work~\citep{BousquettE02,PoggioRMN04,ShwartzSSS10}
that establishes connections between the \emph{stability} of a
learning algorithm and its ability to generalize. Most of these
classic stability notions are defined over some metric on the output
space (rather than on the distribution over outputs), and for these
reasons are also brittle to post-processing and adaptive
composition. The third is the compression-scheme method first
introduced by \cite{LW86} (see, e.g., \cite{shai} for a textbook
introduction). As we show in this paper, the generalization guarantees
that follow from compression schemes \emph{are} robust to
post-processing and adaptive composition. A longstanding conjecture
\citep{W03} states that VC-classes of dimension $d$ have compression
schemes of size $d$, but it is known that boosting \cite{FS97} implies
the existence of a \emph{variable-length} compression scheme that for
any function from a VC-class of dimension $d$ can compress $n$
examples to an empirical risk minimizer defined by a subset of only
$O(d\log n)$ many examples \cite{DMY16}.

A recent line of work \citep{stoc15,nips15,BNSSSU15,RZ15} has studied algorithmic conditions that guarantee the sort of \emph{robust} generalization guarantees we study in this paper, suitable for adaptive data analysis. \cite{stoc15} show that differential privacy (a stability guarantee on the output \emph{distribution} of an algorithm) is sufficient to give robust generalization guarantees, and~\cite{nips15} show that description length bounds on the algorithm's output (i.e., Occam style bounds~\cite{BEHW90}, which have long been known to guarantee standard generalization) are also sufficient.

Differential privacy was introduced by \cite{DMNS06} (see \cite{DR14} for a textbook introduction), and private learning has been a central object of study since \cite{KLNRS08}. The key results we use here are the upper bounds for private learning proven by \cite{KLNRS08} using the exponential mechanism of \cite{MT07}, and the lower bounds for private proper threshold learning due to \cite{BNSV15}. A measure similar to, but distinct from, the notion of \emph{perfect generalization} that we introduce here was briefly studied as a privacy solution concept in \cite{BLR08} under the name ``distributional privacy.''


\section{Preliminaries}\label{sec:prelim}
\subsection{Learning Theory Background}

Let $\cX$ denote a~\emph{domain}, which contains all
possible~\emph{examples}. A~\emph{hypothesis}
$h\colon \cX\rightarrow \{0,1\}$ is a boolean mapping that labels
examples by $\{0, 1\}$, with $h(x) = 1$ indicating that $x$ is a
positive instance and $h(x) = 0$ indicating that $x$ is a negative
instance. A~\emph{hypothesis class} is a set of hypotheses.
Throughout the paper, we elide dependencies on the dimension of the
domain.

We will sometimes write $\cX_L$ for $\cX\times\{0, 1\}$, i.e.,
labelled examples.  Let $\cD_L\in \Delta \cX_L$ be a distribution over
labelled examples; we will refer to it as the {\em underlying
  distribution}. We write $S_L \sim_{i.i.d.} \cD_L^n$ to denote a {\em
  sample} of $n$ labelled examples drawn i.i.d.\ from $\cD_L$.  A
learning algorithm takes such a sample $S_L$ (also known as a {\em
  training set}) as input, and outputs a hypothesis. Note that we use
subscript-$L$ to denote labeling of examples in the domain, in
samples, and in distributions. When $\cD_L$ is well-defined, we also
sometimes write $\cD$ for the marginal distribution of $\cD_L$ over
$\cX$; similarly for $S$ and $S_L$.


Typically, the goal when selecting a hypothesis is to minimize
the~\emph{true error} (also known as the expected error) of the
selected hypothesis on the underlying distribution:
\[
err(h) = \Pr_{(x,y)\sim \cD_L} [h(x) \neq y].
\]
This is in contrast to the {\em empirical error} (also known as the
training error), which is the error of the selected hypothesis $h$ on
the sample $S_L$:
\[
err(S_L, h) \equiv \frac{1}{|S_L|} \sum_{(x_i, y_i)\in S_L}
\mathbf{1}[h(x_i) \neq y_i].
\]

In order to minimize true error, learning algorithms typically seek to
(approximately) minimize their empirical error, and to combine this with a
generalization guarantee, which serves to translate low empirical
error into a guarantee of low true error.


For any set $S\in \cX^n$, let $\cE_S$ denote the empirical
distribution that assigns weight $1/n$ on every observation in
$S$. For any hypothesis $h\colon \cX \rightarrow \{0,1\}$, we will
write $h(\cD)$ to denote $\mathbb{E}_{x\sim \cD} \left[h(x)\right]$
and $h(S)$ to denote
$h(\cE_S) = \mathbb{E}_{x\sim \cE_S} \left[h(x)\right] = 1/n \sum_{x_i
  \in S} h(x_i)$. We say that a hypothesis $h:\cX\rightarrow \{0,1\}$
$\alpha$-{\em overfits} to the sample $S$ taken from $\cD$ if
$|err(h) - err(S_L, h)| \geq \alpha$. Traditional generalization
requires that a mechanism output a hypothesis that does not overfit to
the sample.

\begin{definition}[(Traditional) Generalization]
Let $\cX$ be an arbitrary domain.
A mechanism $\cM\colon \cX_L^n \rightarrow (\cX\rightarrow \{0,1\})$ is
{\em $(\alpha, \beta)$-generalizing} if for all distributions $\cD_L$ over
$\cX_L$, given a sample $S_L\sim_{i.i.d.}\cD_L^n$,
\[
  \Pr\left[\cM(S_L) \text{ outputs } h\colon\cX\rightarrow \{0,1\}
    \text{ such that } |err(h) - err(S_L, h)|\leq \alpha \right]\geq 1-
  \beta,
\]
where the probability is over the choice of the sample $S_L$ and the
randomness of $\cM$.
\label{def:G}
\end{definition}
Note that (traditional) generalization does not prevent $\cM$ from encoding its input sample $S_L$ in the hypothesis $h$ that it outputs. \katrina{The reviewer (quite reasonably) didn't like how we were doing this, so I am skipping the attempt to make the generalization definition also apply in non-learning settings.}

Note that throughout the paper, we focus only on {\em proper learning}, wherein the learner is required to return a hypothesis from the class it is learning, rather than from, e.g., some superset of that class. For simplicity, we frequently omit the word ``proper.''
Within the setting of proper learning, we consider two different models of learning.
In the setting of~\emph{PAC
learning}, we assume that the examples in the support of the underlying distribution are labelled consistently with
some~\emph{target} hypothesis $h^*$ from a known hypothesis class $\cH$. In this case,
we could write $err(h) = \Pr_{x\sim \cD}[h(x) \neq h^*(x)]$.

\begin{definition}[PAC Learning]\label{def:pac}
A hypothesis class $\cH$ over domain $\cX$ is~\emph{PAC learnable}
 if there exists a polynomial $n_{\cH} : \mathbb{R}^2 \rightarrow \mathbb{R}$ and a learning algorithm $\cA$ such that
 for all hypotheses $h^*\in \cH_d$,
 all $\alpha, \beta\in (0, 1/2)$,
 and all distributions $\cD$ over
$\cX$,
 given inputs
$\alpha, \beta$ and a sample $S_L = (z_1, \ldots, z_n)$, where $n \geq n_{\cH}(1/\alpha, \log(1/\beta))$, $z_i = (x_i, h^*(x_i))$ and the $x_i$'s are
drawn $i.i.d.$ from $\cD$, the algorithm $\cA$ outputs a hypothesis
$h\in \cH$ with the following guarantee:
\[
\Pr[err(h) \leq \alpha] \geq 1 - \beta.
\]
The probability is taken over both the randomness of the examples and
the internal randomness of $\cA$. We will say that $\cH$ is PAC
learnable with a learning rate $n_\cH$, and call a learning algorithm
with the above guarantee {\em $(\alpha, \beta)$-accurate}.\sw{Added}
\end{definition}


In the setting of~\emph{agnostic learning}, we do not assume that the labels of the underlying data distribution are consistent with some hypothesis in $\cH$.
The goal then becomes finding a hypothesis whose true error is almost optimal within the hypothesis class $\cH$.

\begin{definition}[Agnostic Learning]\label{def:aglearn}
\emph{Agnostically learnable} is defined identically to PAC learnable with
two exceptions:
\begin{enumerate}
\item the data are drawn and labelled from an arbitrary distribution $\cD_L$ over $\cX\times \{0,1\}$
\item the output hypothesis $h$ satisfies the following
\[
\Pr[err(h) \leq \OPT + \alpha] \geq 1 - \beta,
\]
where $\OPT = \min_{f\in \cH}\{err(f)\}$ and the probability is
taken over both the randomness of the data and the internal randomness
of the algorithm.
\end{enumerate}
\end{definition}

It is known that (in the binary classification setting we study), a hypothesis class is learnable if and only if its \emph{VC-dimension} is polynomially bounded:
\begin{definition}[VC Dimension~\cite{vc}]
A set $S\subseteq \cX$ is~\emph{shattered by} a hypothesis class
$\cH$ if $\cH$ restricted to $S$ contains all $2^{|S|}$ possible
functions from $S$ to $\{0, 1\}$. The~\emph{VC dimension} of $\cH$
denoted $\vc(\cH)$, is the cardinality of a largest set $S$
shattered by $\cH$.
\end{definition}

%

\subsection{Notions of Generalization}

In this section, we introduce the three notions of generalization that are studied throughout this paper.
We say that a mechanism $\cM$ \emph{robustly generalizes} if the mechanism does not
provide information that helps overfit the sample it is given as input. Formally:


%
%

\begin{definition}[Robust Generalization]\label{def.wg}
Let $\cR$ be an arbitrary range and $\cX$ an arbitrary domain. A mechanism $\cM\colon \cX_L^n \rightarrow \cR$ is \emph{$(\eps,
  \delta)$-robustly generalizing} if for all distributions
$\cD_L$ over $\cX_L$ and any adversary $\cA$, with probability $1 - \zeta$
over the choice of sample $S_L\sim_{i.i.d.}\cD_L^n$,
\[
\Pr\left[\cA(\cM(S_L))\text{ outputs } h\colon\cX\rightarrow \{0,1\} \text{ such that } |h(S_L) - h(\cD_L)|\leq \eps  \right]\geq 1- \gamma,
\]
for some $\zeta, \gamma$ such that $\delta=\zeta + \gamma$, where the
probability is over the randomness of $\cM$ and $\cA$.\footnote{Note
  that we do not state the robust generalization guarantee in terms of
  the difference $|err(h') - err(S_L, h')|$ between true error and
  empirical error for some hypothesis $h'$ (as in~\Cref{def:G}), and
  our definition is in fact more general --- in particular, we can let
  $h((x, y)) = \mathbf{1}[h'(x) \neq y]$ to capture the generalization
  notion in terms of error.}
\end{definition}


For our other notions of generalization we require the following
definition of distributional closeness.
\begin{definition}[$(\eps , \delta)$-Closeness]
Let $\cR$ be an arbitrary range, and let $\Delta \cR$ denote the set of all
probability distributions over $\cR$. We say that distributions $\cJ_1, \cJ_2\in \Delta \cR$ are \emph{$(\eps, \delta)$-close} and write $\cJ_1 \approx_{\eps, \delta} \cJ_2$ if for all $\cO\subseteq \cR$,
\[
\Pr_{y\sim \cJ_1}\left[ y \in \cO \right] \leq \exp(\eps) \Pr_{y\sim \cJ_2}\left[y \in \cO\right] + \delta \quad \mbox{ and }\quad
\Pr_{y\sim \cJ_2}\left[y \in \cO \right] \leq \exp(\eps) \Pr_{y\sim \cJ_1}\left[ y\in \cO\right] + \delta.
\]
\end{definition}

\katrina{We initially use notation for the domain, distribution, and sample in the following generalization definitions that is distinct from the analogous notation in learning settings, in order to emphasize that they are meaningful both for general, unlabelled domains, and also in the context of learning (where the domain is over labelled examples).}

Given an arbitrary domain $\cY$, we say samples $T, T'\in \cY^n$ are~{\em neighboring} if they differ on
exactly one element. A mechanism $\cM$ is \emph{differentially private} if the distributions of its outputs are close on neighboring samples.
\begin{definition}[Differential Privacy, \cite{DMNS06}]
A mechanism $\cM\colon \cY^n \rightarrow \cR$ is
\emph{$(\eps, \delta)$-differentially private} if for every pair of
neighboring samples $T,T'\in\cY^n$, $\cM(T) \approx_{\eps, \delta} \cM(T')$.
\end{definition}

\def\Sim{\mbox{\sc Sim}}

Let $\cY$ be an arbitrary domain and $\cR$ be an arbitrary range, and let $\Delta \cY$ denote the set of all probability distributions over
$\cY$.  A \emph{simulator} $\Sim\colon \Delta \cY \rightarrow \cR$ is
a (randomized) mechanism that takes a probability distribution over
$\cY$ as input, and outputs an outcome in the range $\cR$. For any
fixed distribution $\cC\in \Delta\cY$, we sometimes write $\Sim_\cC$
to denote the output distribution $\Sim(\cC)$.

We say that a mechanism $\cM$ \emph{perfectly generalizes} if the distribution of its output when run on a sample is close to that of a simulator that did not have access to the sample.
\begin{definition}[Perfect Generalization]
Let $\cR$ be an arbitrary range and $\cY$ an arbitrary domain.
Let $0\leq \beta < 1$, $\eps\geq 0$, and $0\leq\delta < 1$. A mechanism $\cM\colon \cY^n \rightarrow \cR$ is
\emph{$(\beta, \eps, \delta)$-perfectly generalizing} if for every
distribution $\cC$ over $\cY$ there exists a \emph{simulator} $\Sim_\cC$
such that with probability at least $1-\beta$ over the choice of sample
$T\sim_{i.i.d.}\cC^n$, $\cM(T) \approx_{\eps, \delta} \Sim_\cC$.
\end{definition}



\paragraph{Discussion of the generalization notions}
We will see that all three of the above generalization notions are robust to postprocessing and compatible with adaptive composition,\footnote{Specifically, differentially private algorithms can be adaptively composed in a black box manner, and can be compiled into perfectly generalizing mechanisms (with some loss in their parameters). This gives a recipe for designing perfectly generalizing mechanisms that compose adaptively. Similarly, many methods for guaranteeing robust generalization (including differential privacy, description length bounds, and compression schemes) compose adaptively, giving a recipe for designing robustly generalizing mechanisms that compose adaptively.} making each of them much more appealing than traditional generalization for learning contexts.
Perfect generalization also has an intuitive interpretation as a privacy solution concept that guarantees privacy not just to the individuals in a data sample, but to the sample as a whole (one can think of this as providing privacy to a data provider such as a school or a hospital, when each provider's data comes from the same underlying distribution). Despite the very strong guarantee it gives, we will see that many tasks are achievable under perfect generalization.

\sw{are we using $\cC$ for distributions?}
\katrina{I'm trying to enforce $\cC$ for distributions over $\cY$, which may or may not be labelled, and $\cD$ for distributions over $\cX$, which are unlabelled examples.}

\subsection{Basic Properties of the Generalization Notions}

Here we state several basic properties of the generalization notions defined above.  Proofs are deferred to Appendix \ref{s.prelimproofs}.

The following lemma is a useful tool for bounding the closeness
parameters between two distributions via an intermediate distribution,
such as that of the simulator.  It allows us to say (Corollary
\ref{cor:twin}) that for any perfectly generalizing mechanism, any two
``typical'' samples will induce similar output distributions.

\begin{lemma}
\label{lem.structural} Let $\cJ_1,\cJ_2,\cJ_3$ be distributions over an abstract domain $\cR$.  That is, $\cJ_1,\cJ_2,\cJ_3 \in \Delta \cR$.
If $\cJ_1\approx_{\eps,\delta} \cJ_2$ and
$\cJ_2\approx_{\eps',\delta'} \cJ_3$ where $\eps,\eps'<\ln 2$ then
$\cJ_1\approx_{\eps+\eps',2(\delta + \delta')} \cJ_3$.  If $\delta = \delta'$, then $\cJ_1\approx_{\eps+\eps',3\delta} \cJ_3$.
\end{lemma}


\begin{corollary}
\label{cor:twin}
Suppose that a mechanism $\cM\colon \cY^n \rightarrow \cR$ is
$(\beta, \eps, \delta)$-perfectly generalizing, where $\eps < \ln 2$. Let $T_1,
T_2\sim_{i.i.d.}\cC^n$ be two independent samples. Then with
probability at least $1- 2\beta$ over the random draws of $T_1$ and
$T_2$, the following holds
\[
\cM(T_1) \approx_{2\eps, 3\delta} \cM(T_2).
\]
\end{corollary}

We can show that both robust generalization and perfect generalization
are robust to postprocessing, i.e., arbitrary interpretation.
It is known that differential privacy is also robust to postprocessing~\citep{DR14}.

\begin{lemma}[Robustness to Postprocessing]
\label{lem:post}
Given any $(\alpha, \beta)$-robustly generalizing (resp.
$(\beta, \eps, \delta)$-perfectly generalizing) mechanism
$\cM \colon \cY^n \rightarrow \cR$ and any post-processing procedure
$\cA\colon \cR \rightarrow \cR'$, the composition
$\cA\circ\cM \colon \cY^n \rightarrow \cR$ is also
$(\alpha, \beta)$-robustly generalizing (resp.
$(\beta, \eps, \delta)$-perfectly generalizing).
\end{lemma}

Theorem \ref{thm.basic} says that the composition of multiple
$(\beta, \eps, 0)$-perfectly generalizing mechanisms is also perfectly
generalizing, where the $\beta$ and $\eps$ parameters ``add up''.

\begin{theorem}[Basic Composition]
\label{thm.basic}
Let $\cM_i \colon \cY^n\rightarrow \cR_i$ be
$(\beta_i, \eps_i, 0)$-perfectly generalizing for $i=1, \ldots, k$.
The composition
$\cM_{[k]} \colon \cY^n \rightarrow \cR_1 \times \cdots \times \cR_k$,
defined as $\cM_{[k]}(T) = \left( \cM_1(T), \ldots \cM_k(T) \right)$
is $(\sum_{i=1}^k \beta_i, \sum_{i=1}^k \eps_i, 0)$-perfectly
generalizing.
\end{theorem}

A very recent work by~\cite{BF16} studies the notion of~\emph{typical
  stability}, which generalizes perfect generalization. In particular,
a mechanism is perfectly generalizing if it is typically stable with
respect to product distributions $\cD^n$. They show the class of
typically stable mechanisms is closed under adaptive composition,
implying an adaptive composition theorem for perfectly generalizing
mechanisms.\footnote{ A previous version of our paper contained an
  error in the proof of the adaptive composition theorem for perfect
  generalization. We are grateful to Raef Bassily and Adam Smith for
  bringing this to our attention.\sw{changed}}

\sw{added the above; please check}

\section{Robust Generalization via Compression Schemes}
In this section, we present a new technique, based on the idea
of~\emph{compression bounds}, for designing learning algorithms with
the robust generalization guarantee. Recent work
\citep{stoc15,nips15,BNSSSU15} gives two other techniques for
obtaining robust generalizing mechanisms.  As we will see, our new
technique allows one to learn hypothesis classes under robust
generalization for which the two previous techniques do not
apply. More surprisingly, we show that any PAC/agnostically learnable
hypothesis class can also be learned under robust generalization with
nearly optimal sample complexity.

We first give a definition for what it means to learn a
hypothesis under robust generalization.
\begin{definition}[RG PAC/Agnostic Learning]\label{def:rglearnable}
  A hypothesis class $\cH$ over domain $\cX$ is {\em PAC/agnostically
  learnable under robust generalization (RG-PAC/agnostically
  learnable)} if there exists a polynomial $n_\cH\colon
  \RR^4\rightarrow \RR$ and a learning algorithm $\cA$ such that for
  all $\alpha, \beta, \eps, \delta \in (0, 1/2)$, given inputs
  $\alpha, \beta, \eps, \delta$ and a sample $S_L\in \cXl^n$ where
  $n\geq n_\cH(1/\alpha, 1/\eps, \log(1/\beta), \log(1/\delta))$, the
  algorithm $\cA$ is an $(\alpha, \beta)$-accurate PAC/agnostic
  learner, and is $(\eps, \delta)$-robustly generalizing.
\end{definition}

\subsection{Compression Learners}
For any function $k \colon \NN \rightarrow \NN$, we say that a
hypothesis class has a compression scheme of size $k$ if any arbitrary
set $S_L$ of $n$ labelled examples can be mapped to a~\emph{sequence}
of $k(n)$ input examples, from which it is possible to compute an
empirical risk minimizer for $S_L$. 
\begin{definition}[Compression Scheme~\cite{LW86}]
  Let $\cH$ be a hypothesis class and let
  $k\colon \NN\rightarrow \NN$. We say that $\cH$ has
  a~\emph{compression scheme} of size $k$ if for all $n\in \NN$, there
  exists an integer $k'\leq k(n)$, a~\emph{compression algorithm}
  $A \colon \cX_L^n \rightarrow \cX_L^{k'}$ and an~\emph{encoding
    algorithm} $B\colon \cX_L^{k'} \rightarrow \cH$ such that for any
  arbitrary set $S_L$ of $n$ labelled examples, $A$ will select a
  sequence of examples
  $A(S_L) = (z_{i_1}, z_{i_2}, \ldots, z_{i_{k'}})\in S_L^{k'}$, and
  $B$ will output a hypothesis $h' = B(A(S_L))$ that is
  an~\emph{empirical risk minimizer}; i.e.
  $err(S_L, h') \leq err(S_L, h)$ for all $h\in \cH$. We will call
  the algorithm $\cL = B\circ A$ a compression learner of size $k$ for
  the hypothesis class $\cH$.\footnote{Note that this definition of
    variable-length compression scheme (where the number of examples
    output by the compression algorithm depends on the input sample
    size) is more general than the one defined in~\cite{LW86}.}
 
\end{definition}

\begin{remark}{A natural extension to the compression scheme defined
    above is \emph{approximate} compression schemes \cite{DMY16},
    which produce approximate empirical risk minimizers rather than
    exact empirical risk minimizers. A particularly simple and naive
    approximate compression scheme results from subsampling: since it
    is possible to produce an $\eps$-approximate empirical risk
    minimizer for any function drawn from a VC-class of dimension $d$
    using $O(d/\eps^2)$ samples, it immediately follows that every
    VC-class of dimension $d$ admits an $\eps$-approximate compression
    scheme of size $k = O(d/\eps^2)$. As we will see, such a compression
    scheme is in general quite inefficient in terms of sample
    complexity, and by using a more sophisticated boosting-based
    compression scheme \cite{DMY16}, it is possible to obtain
    robust generalization with nearly optimal sample complexity for
    every VC-class.
}\end{remark}\sw{added}

Next, we want to show that any compression learner of small size
satisfies robust generalization. As an intermediate step, we recall
the following result, which follows from a standard application of a
concentration bound.

\begin{lemma}[see, e.g.,~\cite{shai} Theorem 30.2]\label{lem:shai}
  Let $n, k'\in \NN$ such that $n\geq 2k'$. Let
  $A\colon \cX_L^n \rightarrow \cX_L^{k'}$ be an algorithm that takes
  a sample $S_L$ of $n$ labelled examples as input, and selects a
  sequence of labelled examples
  $A(S_L) = (z_{i_1}, z_{i_2}, \ldots, z_{i_{k'}})\in S_L^{k'}$ of
  length ${k'}$.  Let algorithm
  $B\colon \cXL^{k'} \rightarrow (\cX\rightarrow \{0,1\})$ take a
  sequence of ${k'}$ labelled examples and return a hypothesis.

  For any random sample $S_L\sim_{i.i.d.}\cD_L^n$, let
  $V_L = \{z\mid z\notin A(S_L)\}$ be the set of examples not selected by
  $A$, and write $V$ for the unlabelled version of $V_L$.  Let
  $h= B(A(S_L))$ be the hypothesis output by $B$. Then, with
  probability of at least $1-\delta$ over the random draws of $S_L$
  and the randomness of $A$ and $B$, we have
\[
| h(V) - h(\cD) | \leq  \sqrt{h(V) \frac{4{k'}\log(2n/\delta)}{n}} +
\frac{8{k'}\log(2n/\delta)}{n}
\]
Recall that $h(\cD) = \E_{x\sim\cD}\left[h(x)\right]$ denotes the
expected value of $h$, and $h(V) = \frac{1}{n-{k'}}\sum_{x\in V} h(x)$
is the average value of $h$ over the examples in $V$.
\end{lemma}
This theorem is useful in analyzing the guarantees of a compression
learner. If we interpret $A$ as a compression algorithm, and $B$ as an
encoding algorithm that outputs a hypothesis $h$, Lemma \ref{lem:shai}
says that the empirical error of $h$ over the remaining subset $V$ is
close to its true error.

However, we can also interpret algorithm $B$ as an adversary who is
trying to overfit a hypothesis to the input sample $S_L$. Since the
hypothesis output by a compression algorithm is uniquely determined by
the sequence of examples output by the compression algorithm $A$, we
could think of the adversary post-processing the size-${k'}$ sequence of
examples that defines the output hypothesis. Therefore, it suffices to
show that the compression algorithm $A$ is robustly generalizing. We
will establish this by showing that any algorithm that outputs a small
sequence of the input sample is robustly generalizing:

\begin{lemma}\label{thm:cats}
  Let $n,{k'}$ be integers, $\eps, \delta > 0$, and let
  $A \colon \cXl^n \rightarrow \cXl^{k'}$ be an algorithm that takes
  any set $S_L\in\cXl^n$ as input and outputs a sequence
  $T \in S_L^{k'}$ of size ${k'}$. Then $A$ is
  $(\eps, \delta)$-robustly generalizing for
\[
  \eps = \sqrt{\frac{{4 k'}\log(n/\delta)}{n}} + \frac{8 k'
    \log(2n/\delta)}{n} + \frac{k'}{n}.
\]
\end{lemma}

\begin{proof}
We will appeal to \Cref{lem:shai}.
Let $F\colon \cXl^{k'} \rightarrow \{\cX \rightarrow \{0,1\}\}$ be a
deterministic mapping from samples of size ${k'}$ to hypotheses. Let
$S_L\sim_{i.i.d.}\cD^n$ be a random sample of size $n$,  $T=A(S_L)$
be the sequence output by the compression algorithm, $V$ be the
examples (without labels) not selected by $A$, and $f = F(T)$ be the
function output by the adversary. By the result of Lemma \ref{lem:shai}, we
know that with probability at least $1-\delta$ over the random draws
of $S_L$, the following holds,
\begin{equation*}
|f(V) - f(\cD)| \leq \sqrt{\frac{4{k'}\log(2n/\delta)}{n}} + \frac{8{k'}\log(2n/\delta)}{n} \equiv C
\end{equation*}
Let $S$ be the examples in $S_L$ but without labels. By the triangle
inequality we have
\begin{align}
|f(S_L) - f(\cD)| &\leq \frac{1}{n} \left|\sum_{z\in S_L} (f(z) - f(\cD)) \right| \nonumber \\
& =\frac{1}{n} \left|\sum_{z\in V} (f(z) - f(\cD)) + \sum_{z\notin V}(f(z) - f(\cD))\right| \nonumber \\
& \leq\frac{1}{n} \left|\sum_{z\in V} (f(z) - f(\cD))\right| + \frac{1}{n}\left|\sum_{z\notin V}(f(z) - f(\cD))\right| \nonumber \\
&\leq \frac{C \, n}{n} + \frac{{k'}}{n} = C + \frac{{k'}}{n}\label{puppy}
\end{align}
which recovers our stated bound.
\end{proof}

Now we are ready to show that any hypothesis class that admits a
compression scheme of small size is learnable under robust
generalization.

\begin{theorem}[Compression implies RG Learnability]\label{thm:compressRG}
  Let $\cH$ be a hypothesis class with a compression scheme of size
  $k\colon \NN\rightarrow \NN$, and let $\cA\colon \cH\rightarrow \cH$
  be any adversary. Then given any input sample
  $S_L \sim_{i.i.d.} \cD_L^n$ of size $n$, the compression learner
  $\cL$ for $\cH$ outputs an hypothesis $h$ such that with probability
  at least $1 - \delta$, the error satisfies
  $err(h) \leq \min_{h'}err(h') + \eps$, and the adversary outputs a
  hypothesis $f = \cA(h)$ that satisfies $|f(S_L) - f(\cD)| \leq \eps$
  with
  \[
    \eps = O\left(\sqrt{\frac{k(n) \log(n/\delta)}{n}}\right),
  \]
  as long as $n \geq 8{k(n)}\,  \log(2n/\delta)$.
\end{theorem}

\begin{proof}
  Note that when $n \geq 8{k(n)}\, \log(2n/\delta)$, the bound on
  $\eps$ in Lemma \ref{thm:cats} becomes
  $O\left(\sqrt{\frac{k(n) \log(n/\delta)}{n}}\right)$. Then by
  applying Lemma \ref{thm:cats}, we can guarantee that
  $|f(S_L) - f(\cD)|\leq \eps$ with probability at least $1 -
  \delta$. Then the accuracy guarantee of the learner's output
  hypothesis directly follows by setting $\cA$ to be the identity map.
\end{proof}

We can also show that compression learners continue to give robust
generalization under adaptive composition.

\begin{theorem}[Adaptive Composition for Compression Learners]\label{thm:comcomp}
  Let $\cM_{[m]}\colon \cX^n \rightarrow \cH^m$ be an adaptive
  composition of compression schemes such that for any $S\in \cX^n$,
  $\cM_{[m]}(S) = (h_1, \ldots , h_m)$, where $h_1 = \cM_1(S)$,
  $h_2 = \cM_2(S;h_1), \ldots, h_m = \cM(S; h_1, \ldots ,h_{m-1})$,
  where $\cM_i(\cdot; h_1, \ldots , h_{i-1})$ is a compression learner
  of size $k_i$ for all choices of $h_1, \ldots , h_{i-1}$.  Let
  $k = \sum_{i=1}^m k_i$. Then $\cM_{[m]}$ is
  $(\eps, \delta)$-robustly generalizing, where
\[
\eps = O\left(\sqrt{\frac{k\log(n/\beta)}{n}}\right),
\]
as long as $n \geq 8k\log(2n/\beta)$.
\end{theorem}

\begin{proof}
For each $\cM_i$, we can write it as $\cM_i(\cdot; h_1, \ldots ,
h_{i-1}) = (B_i\circ A_i)$, where $A_i$ is the compression algorithm
and $B_i$ is the encoding algorithm.  Note that the sequence of output
hypotheses is just a postprocessing of the sequence of examples
output by the compression algorithms---that is, given the sequence
of examples output by the compression algorithms, we can uniquely
determine the set of output hypotheses. So it suffices to prove that
the adaptive composition of the compression algorithms satisfies robust
generalization. Note that the composed compression algorithms can
be viewed as a single compression algorithm that releases a sequence
of examples of length $k$. By directly applying Lemma \ref{thm:cats}, we
recover the stated bound.
\end{proof}

\subsection{Robust Generalization via Differential Privacy and Description Length}

We  briefly review two existing techniques for obtaining algorithms with
robust generalization guarantees, from the recent line of work starting
with~\cite{stoc15}, and followed by~\cite{nips15, BH15,
  BNSSSU15}. Here we will rephrase their results in terms of robust
generalization (this terminology is new to the present paper).

First, it is known that differential privacy implies
robust generalization.
\begin{theorem}[\cite{BNSSSU15}]
Let $\cM\colon \cX_L^n\rightarrow \cR$ be a $(\eps,
\delta)$-differentially private mechanism for $n\geq
O(\ln(1/\delta)/\eps^2)$. Then $\cM$ also satisfies $(O(\eps),
O(\delta/\eps))$-robust generalization.
\end{theorem}

Algorithms with a small output range (i.e., each output can be
described using a small number of bits) also enjoy
robust generalization.
\begin{theorem}[\cite{nips15}]\label{thm:description}
Let $\cM\colon \cX_L^n \rightarrow \cR$ be a mechanism such that $|\cR|$
is bounded. Then $\cM$ satisfies $(\alpha, \beta)$-robust generalization, with
$\alpha = \sqrt{\frac{\ln(|\cR|/\beta)}{2n}}$.
\end{theorem}

\subsection{Case Study: Proper Threshold Learning}\label{sec:thre}
Next, we consider the problem of properly learning~\emph{thresholds}
in the PAC setting. We will first note that when the domain size is
infinite, there is no proper PAC learner that is differentially
private or has finite output range. In contrast to these impossibility
results, we show that the class of threshold functions admits a simple
compression scheme, and hence a PAC learning algorithm that satisfies
robust generalization. This result, in particular, gives a
separation between the power of learning under robust generalization
and that of learning under differential privacy.

Consider the hypothesis class of~\emph{threshold functions}
$\{h_x\}_{x\in \cX}$ over a totally ordered domain $\cX$, where
$h_x(y) = 1$ if $y\leq x$ and $h_x(y) = 0$ if $y > x$. We will first
recall an impossibility result for privately learning thresholds.

\begin{theorem}[\cite{BNSV15} Theorem 6.2]\label{bun}
Let $\alpha > 0$ be the accuracy parameter (as in~\Cref{def:pac}). For
every $n\in \NN$, and $\delta \leq 1/(1500n^2)$, any $(1/2,
\delta)$-differentially private and $(\alpha, 1/8)$-accurate (proper)
PAC learner for threshold functions requires sample complexity $n =
\Omega\left(\log^*|\cX|/\alpha\right)$.
\end{theorem}

In particular, the result of~\Cref{bun} implies that there is no
private proper PAC learner for threshold functions over an infinite
domain. Similarly, we can show that there is no proper PAC learner for
thresholds that has a finite outcome range. \katrina{any objection
  that I switched to ``range'' over ``space'' throughout?}\sw{sure}

\begin{lemma}
  Let $\cH$ be the hypothesis class of threshold functions. For any
  $n\in \NN$ and any learner $\cM\colon \cX_L^n \rightarrow \cH'$ such
  that the output hypothesis class $\cH'$ is a subset of $\cH$ and has
  bounded cardinality, there exits a distribution $\cD\in \Delta\cX$
  such that the output hypothesis has true error $err(h) \geq 1/2$.
\end{lemma}

\begin{proof}
Let $|\cH'| = m$. We can write $\cH' = \{h_{x_1}, h_{x_2}, \ldots,
h_{x_m}\}$ such that $x_1 < x_2 < \ldots < x_m$. Let $y, z$ be points
such that $x_1 < y < z < x_2$. Let $\cD$ be a distribution over $\cX$
that puts half of the probability mass on $y$ and the other half on
$z$. Suppose our target hypothesis is $c = h_y$. Then $c(y) = 1$ and
$c(z) = 0$. Note that for each $h\in \cH'$, it must be case that $h(y)
= h(z)$, and thus its true error must be at least $1/2$.
\end{proof}

Now we will show that the class of threshold functions can be properly
PAC learned under the constraint of robust generalization~\emph{even
  when} the domain size is infinite.

\begin{theorem}
  Let $\cH$ be the hypothesis class of threshold functions. There
  exists a compression learner for $\cH$ such that when given a input
  sample $S_L\sim_{i.i.d.} \cD_L^n$ of size $n$, it is both
  $(\eps, \delta)$-accurate and $(\eps, \delta)$-robustly generalizing
  for any $\delta \in (0, 1)$ and
\[
  \eps = O\left( \sqrt{\frac{\log(n/\delta)}{n}} \right)
\]
as long as $n\geq 8 \log(2n/\delta)$.
\end{theorem}

\begin{proof}
  Consider the compression function $A$, that, given a sample, outputs
  the largest positive example $s_+$ in the sample. Then consider the
  encoding function $B$ that, given any example $s_+$, returns the
  threshold function $h_{s_+}$. Such a threshold function will
  correctly label all the examples in the sample. This gives us a
  compression scheme of size 1 for the class of threshold
  functions. Then the result follows directly
  from~\Cref{thm:compressRG}.
\end{proof}

\subsection{Every Learnable Class is Learnable under Robust Generalization}
Finally, we will show that any PAC-learnable hypothesis class can be
learned under robust generalization with a logarithmic blowup in the
sample complexity. We will rely on a result due to~\cite{DMY16}, which
shows that any learnable class admits a compression scheme of size
scaling logarithmically in the input sample size $n$.

\begin{theorem}[\cite{DMY16} (see Theorem 3.1)]\label{thm:dmv}
  Let $\cH$ be a hypothesis class that is PAC/agnostically learnable
  with VC-dimension $d$; then it has a compression scheme of size
  $$k(n) = O(d \log(n) \log\log(n) + d \log(n) \log(d)).$$
\end{theorem}

Our result then follows directly from~\Cref{thm:compressRG}
and~\Cref{thm:dmv}.

\begin{theorem}\label{thm:all}
  Let $\cH$ be a hypothesis class. Suppose that $\cH$ is
  PAC/agnostically learnable with $\vc(\cH) = d$. Then there exists a
  compression learner for $\cH$ such that when given input sample
  $S_L\sim_{i.i.d.}\cD_L^n$, $\cL$ is both $(\eps, \delta)$-accurate and
  $(\eps, \delta)$-robustly generalizing for any $\delta\in (0, 1)$
  and $\eps = \tilde O\left(\sqrt{ {d}/{n}} \right)$ as long as
  $n \geq 16 d\log(d) \log^3(n/\delta)$.
\end{theorem}

\begin{remark}\label{rm:subsam}
  Note that we can obtain a similar result with the approximate
  compression scheme of subsampling.  In particular, for every
  VC-class of dimension $d$, the compression learner that uses
  subsampling as its compression algorithm is both
  $(\eps,\delta)$-accurate and $(\eps,\delta)$-robustly generalizing
  with:
$$\eps = O\left(\left(\frac{d \log(n/\delta)}{n}\right)^{1/4}\right)$$
which is polynomial, but is quadratically suboptimal.
\end{remark}



\section{Learning under Perfect Generalization}

In this section, we will focus on the problem of agnostic learning
under the constraint of perfect generalization. Our main result gives
a perfectly generalizing generic learner in the settings where the
domain $\cX$ or the hypothesis class $\cH$ has bounded size. The
sample complexity will depend logarithmically on these two quantities.
Furthermore, we give a reduction from any perfectly generalizing
learner to a differentially private learner that preserves the sample
complexity bounds (up to constant factors). This allows us to carry
over lower bounds for differentially private learning to learning
under perfect generalization. In particular, we will show that proper
threshold learning with unbounded domain size is impossible under
perfect generalization.

We will first define what it means to learn a hypothesis
under perfect generalization.

\begin{definition}[PG PAC/Agnostic Learning]\label{pglearnable}
  A hypothesis class $\cH$ over domain $\cX$ is PAC/agnostically
  learnable under perfect generalization (PG-PAC/agnostically
  learnable) if there exists a polynomial
  $n_\cH\colon \RR^5\rightarrow \RR$ and a learning algorithm $\cA$
  such that for all $\alpha,\gamma, \beta, \eps, \delta \in (0, 1/2)$,
  given inputs $\alpha,\gamma, \beta, \eps, \delta$ and a sample
  $S_L\in \cXl^n$ where
  $n\geq n_\cH(1/\alpha, 1/\eps,\log(1/\gamma), \log(1/\beta),
  \log(1/\delta))$, the algorithm $\cA$ is an
  $(\alpha, \gamma)$-accurate PAC/agnostic learner, and is
  $(\beta, \eps, \delta)$-perfectly generalizing.
\end{definition}

\subsection{Generic PG Agnostic Learner}\label{sec:generic_pg_learner}
Now we present a generic perfectly generalizing agnostic learner,
which is based on the~\emph{exponential mechanism} of~\cite{MT07} and analogous to the generic learner of \cite{KLNRS08}.

Our learner, formally presented in~\Cref{alg.generic}, takes
generalization parameters $\eps, \beta$, a sample of $n$ labelled
examples $S_L\sim_{i.i.d.}  \cD_L^n$, and a hypothesis class $\cH$ as
input, and samples a random hypothesis with probability that
is exponentially biased towards hypotheses with small empirical error.
We show that this algorithm is perfectly generalizing.

\begin{algorithm}[h!]
  \caption{Generic Agnostic Learner $\cA$($\beta$, $\eps$, $S_L$, $\cH$)}
  \label{alg.generic}
  \begin{algorithmic}
 \STATE{\textbf{Output} $h\in\cH$ with probability proportional to $\exp\left(\frac{ -\sqrt{|S_L|}\cdot \eps \cdot
  err(S_L, h)}{\sqrt{2\ln(2|\cH|/\beta)}}\right)$}
\end{algorithmic}
\end{algorithm}

\begin{lemma}\label{lem.expissg}
Given any $\eps, \beta > 0$ and finite hypothesis class $\cH$, the
learning algorithm $\cA(\beta,\eps, \cdot, \cdot)$ is $(\beta, \eps,
0)$-perfectly generalizing.
\end{lemma}

\begin{proof}
Let $S_L\sim_{i.i.d.}\cD_L^n$ be a labelled random sample of size $n$. Note that
since each $(x_i, y_i)$ in $S_L$ is drawn from the underlying
distribution $\cD_L$, we know that for each $h\in \cH$,
\[
\E_{S_L\sim_{i.i.d.}\cD_L^n} [err(S_L, h)] = err(h).
\]

 Fix any $h\in \cH$.  Then by a Chernoff-Hoeffding bound, we know that with
 probability at least $1 - \beta/|\cH|$, the following holds:
\begin{equation}
\left| err(S_L, h) - err(h) \right| \leq \sqrt{\frac{1}{2n}\ln\left(\frac{2 |\cH|}{\beta} \right)}.
\label{eq:error}
\end{equation}
Applying a union bound, we know that the above holds for all $h\in
\cH$ with probability at least $1 - \beta$. We will condition on
this event for the remainder of the proof. Now consider the following
randomized simulator:
\[
\Sim(\cD_L): \mbox{ output } h\in\cH \mbox{ with probability proportional to } \exp\left(\frac{-\eps \cdot \sqrt{n} \cdot err(h)}{\sqrt{2\ln(2|\cH|/\beta)}}\right).
\]
We want to show that the output distributions satisfy $\cA(\beta,
\eps, S_L)\approx_{\eps,0} \Sim(\cD_L)$, where $S_L\sim_{i.i.d.}\cD_L^n$ is a
labelled random sample of size $n$. Let $Z = \sum_{h\in \cH}
\exp\left(\frac{-\eps \sqrt{n}\cdot
  err(S_L,h)}{\sqrt{2\ln(2|\cH|/\beta)}} \right)$ and $Z' =
\sum_{h\in \cH} \exp\left(\frac{-\eps \cdot \sqrt{n}\cdot
  err(h)}{\sqrt{2\ln(2|\cH|/\beta)}} \right)$. For each $h\in
\cH$,
\begin{align*}
\frac{\Pr[\cA(\beta, \eps, S_L, \cH) = h]}{\Pr[\Sim(\cD_L) = h]} &=
\frac{\exp\left(\frac{-\eps \cdot \sqrt{n} \cdot
    err(S_L,h)}{\sqrt{2\ln(2|\cH|/\beta)}} \right)/Z}
{\exp\left(\frac{-\eps \cdot \sqrt{n}\cdot err(h)}{\sqrt{2\ln(2|\cH|/\beta)}}
\right)/ Z'}\\
&= \exp\left(\frac{\eps \cdot \sqrt{n}\left(err(h) -  err(S_L, h) \right)}{\sqrt{2\ln(2|\cH|/\beta)}}\right) \cdot \frac{Z'}{Z}\\
&\leq \exp\left(\frac{\eps}{2}\right) \exp\left(\frac{\eps}{2}\right) \cdot \frac{Z}{Z} \\
&= \exp(\eps).
\end{align*}
A symmetric argument would also show that $\frac{\Pr[\Sim(\cD_L) =
    h]}{\Pr[\cA(\beta, \eps, S_L, \cH) = h]} \leq \exp(\eps)$. Therefore, $\cA(\beta, \eps, \cdot, \cdot)$ is $(\beta, \eps,
0)$-perfectly generalizing.
\end{proof}

\rc{Double check this statement once we've settled on a way to incorporate $d$.}

\begin{theorem}\label{thm:mainlearning}
Let $\cH$ be a finite hypothesis class and $\alpha, \gamma>0$. Then
the generic learner~\Cref{alg.generic} instantiated as $\cA(\gamma,
\eps,\cdot, \cH)$ is $(\alpha, \gamma)$-accurate as long as the sample
size
\[
n\geq \frac{6}{\eps^2 \alpha^2}\left(\ln(2|\cH|) +
\ln(1/\gamma)\right)^3 .
\]
\end{theorem}

\begin{proof}
Let $S_L \sim_{i.i.d.} \cD_L^n$, and let the algorithm $\cA(\gamma,
\eps, S_L, \cH)$ be the Generic Agnostic Learner of Algorithm
\ref{alg.generic}. Consider the event $E = \{\cA(\gamma, \eps, S_L, \cH)
= h \; | \; err(h) > \alpha + \OPT\}$, where $\alpha$ is our target
accuracy parameter. We want to show that $\Pr[E] \leq \gamma$ as long
as the sample size $n$ satisfies the stated bound.

By a Chernoff-Hoeffding bound (similar to~\Cref{eq:error}), we have
that with probability at least $1 - \gamma/2$, the following condition
holds for each $h\in \cH$:
\[
\left| err(S_L, h) - err(h) \right| \leq \sqrt{\frac{1}{2n}\ln\left(\frac{4 |\cH|}{\gamma} \right)} \equiv B(n).
\]
We will condition on the event above. Let $h^* = \arg\min_{h'\in\cH}
err(h')$ and let $\OPT = err(h^*)$, then
\[
\min_{h'\in \cH} err(S_L, h') \leq err(S_L, h^*) \leq err(h^*) + B(n) = \OPT + B(n)
\]

Recall that for each $h\in \cH$, the probability that the hypothesis
output by $\cA(\gamma, \eps, S_L, \cH)$ is $h$ is,
\begin{align*}
&\quad  \frac{\exp\left(-\eps \sqrt{n} \cdot err(S_L, h) / \sqrt{2\ln(2|\cH|/\gamma)}\right)}{\sum_{h'\in \cH}\exp\left(-\eps \sqrt{n} \cdot err(S_L, h') / \sqrt{2\ln(2|\cH|/\gamma)}\right)}\\
 &\leq \frac{\exp\left(-\eps \sqrt{n} \cdot err(S_L, h) / \sqrt{2\ln(2|\cH|/\gamma)}\right)}{\max_{h'\in \cH}\exp\left(-\eps \sqrt{n} \cdot err(S_L, h') / \sqrt{2\ln(2|\cH|/\gamma)}\right)}\\
  &=  \exp\left(-\eps \sqrt{n} \cdot (err(S_L, h) -\min_{h'\in \cH}err(S_L, h')) / \sqrt{2\ln(2|\cH|/\gamma)}\right)\\
&\leq   \exp\left(-\eps \sqrt{n} \cdot (err(S_L, h) -\OPT - B(n)) / \sqrt{2\ln(2|\cH|/\gamma)}\right).
\end{align*}

Taking a union bound, we know that the probability that $\cA(\gamma,
\eps, S_L, \cH)$ outputs a hypothesis $h$ with empirical error $err(S_L,
h) \geq \OPT + 2B(n)$ is at most $|\cH|\exp\left(-\eps \sqrt{n}B(n)/\sqrt{2
  \ln(2|\cH|/\gamma)} \right)$.

Set $B(n) = \alpha / 3$, and the event $E = \{\cA(\gamma, \eps, S_L, \cH_d) =
h \; | \; err(h) > \alpha + \OPT\}$ implies
\[
err(S_L, h) \geq \OPT + 2\alpha/3 = \OPT + 2B(n) \quad \mbox{ or }
\quad |err(S_L, h) - err(h)| \geq \alpha/3 = B(n).
\]
It is sufficient to set $n$ large enough to bound the probabilities of
these two events. Further if we a sample size $n \geq \frac{6}{\eps^2
  \alpha^2} \left(\ln(2|\cH|/\gamma) \right)^3$, both probabilities
are bounded by $\gamma/2$, which means we must have $\Pr[E]\leq \gamma$.
\end{proof}

\subsection{PG Learning with VC Dimension Sample Bounds}


We can also extend the sample complexity bound
in~\Cref{thm:mainlearning} to one that is dependent on the VC-dimension of the hypothesis class $\cH$, but resulting bound will have
a logarithmic dependence on the size of the domain $|\cX|$.

\begin{corollary}
Every hypothesis class $\cH$ with finite VC dimension is PG
agnostically learnable with a sample size of $n = O\left((\vc(\cH)
\cdot \ln|\cX| + \ln{\frac{1}{\beta}})^3 \cdot \frac{1}{\eps^2\alpha^2}
\right)$.
\end{corollary}

\begin{proof}
  By Sauer's lemma (see e.g.,~\cite{KV94}), we know that there are at
  most $O(|\cX|^{\vc(\cH)})$ different labelings of the domain
  $\cX$ by the hypotheses in $\cH$. We can run the exponential
  mechanism over such a hypothesis class $\cH'$ with cardinality $|\cH'| =
  O\left( |\cX|^{\vc(\cH)} \right)$. The complexity bound follows
  from~\Cref{thm:mainlearning} directly.
\end{proof}

\subsection{Limitations of PG learning} 
We have so far given a generic agnostic learner with perfect
generalization in the cases where either $|\cX|$ or $|\cH|$ is finite.
We now show that the finiteness condition is necessary, by revisiting
the threshold learning problem in~\Cref{sec:thre}. In particular, we
will show that when both of the domain size and hypothesis class are
infinite, properly learning thresholds under perfect generalization is
impossible. Our result crucially relies on a reduction from a
perfectly generalizing learner to a differentially private learner,
which allows us to apply lower bound results of differentially private
learning(such as~\Cref{bun}) to PG agnostic learning.  \sw{Todo:
  Haven't changed some of the $S$ to $S_L$} \rc{I think I got them
  all}

First, let's consider the reduction in~\Cref{alg.sgtodp}, which is a
black-box mechanism that takes as input a perfectly generalizing
mechanism $\cM \colon \cX_L^n \rightarrow \cR$ and a labelled sample $S_L\in
\cX_L^n$, and outputs an element of $\cR$.  We show that this new
mechanism $\cM'(\cM,\cdot)$ is differentially private.


\begin{algorithm}[h!]
 \caption{$\cM'$($\cM \colon \cX_L^n \rightarrow \cR$, $S_L\in \cX_L^n$)}
 \label{alg.sgtodp}
 \begin{algorithmic}
 \STATE{Let $\cE_{S_L}$ be the empirical distribution that assigns weight $1/n$ to each of the data points in $S_L$}
 \STATE{Sample $T_L\sim_{i.i.d.} (\cE_{S_L})^n$}
 \STATE{\textbf{Output} $\cM(T_L) \in \cR$}
\end{algorithmic}
\end{algorithm}


\begin{theorem}\label{thm.sgtodp}
Let $\beta < 1/2e$ and $\eps \leq \ln(2)$, and $\cM$ be a $(\beta,
\eps, \delta)$-perfectly generalizing mechanism, then the mechanism
$\cM'(\cM, \cdot)$ of Algorithm \ref{alg.sgtodp} is $(4\eps, 16\delta
+ 2\beta)$-differentially private.
\end{theorem}

\begin{proof}
Let $S_L, S_L'\in \cX^n$ be neighboring databases that differ on the $i$th
entry, and let $\cE_{S_L}$ and $\cE_{S_L'}$ denote their corresponding
empirical distributions.  Since $\cM$ is $(\beta, \eps,
\delta)$-perfectly generalizing, there exists a simulator $\Sim$ such
that with probability at least $1-\beta$ over choosing $T_L\sim_{i.i.d.}
(\cE_{S_L})^n$,
\begin{equation}\label{eq:simx}
\cM(T_L) \approx_{\eps, \delta} \Sim.
\end{equation}
Similarly, there exists a simulator $\Sim'$
such that with probability at least $1-\beta$ over choosing
$T_L'\sim_{i.i.d.} (\cE_{S_L'})^n$,
\begin{equation}\label{eq:simx1}
\cM(T_L') \approx_{\eps, \delta} \Sim'.
\end{equation}

Let
$R_1 = \{T_L\in \cX_L^n \mid \cM(T_L)\approx_{\eps, \delta} \Sim\}$ and
$R_2 = \{T_L'\in \cX_L^n \mid \cM(T_L')\approx_{\eps, \delta} \Sim'\}$.
We want to first show that there exists a dataset $T^*_L$ such that
$T_L^*\in R_1$ and $T_L^*\in R_2$.

Let $\{(x_i, y_i)\} = S_L \setminus S_L'$ and let
$R_3 = \{T_L\in \cX_L^n \mid (x_i, y_i)\notin T_L\}$.
\[
  \Pr_{T_L \sim_{i.i.d.}(\cE_{S_L})^n}[T_L \in R_3] = \Pr_{T_L'
    \sim_{i.i.d.}(\cE_{S_L'})^n}[T_L' \in R_3] = (1 - 1/n)^n \approx
  1/e.
\]
Moreover, for any $T\in R_3$, 
\[
  \Pr_{T_L \sim_{i.i.d.}(\cE_{S_L})^n}[T_L = T ] = \Pr_{T_L'
    \sim_{i.i.d.}(\cE_{S_L'})^n}[T_L' = T]
\]
Note that any dataset $T_L$ in $R_3$ also lies in the supports of both
$(\cE_{S_L})^n$ and $(\cE_{S_L'})^n$. It follows that
\begin{align*}
  &  \Pr_{T_L \sim_{i.i.d.}(\cE_{S_L}^n)}[ T_L \in \left(R_1 \cap
    R_2\right) ]\\
  \geq &\Pr_{T_L \sim_{i.i.d.}(\cE_{S_L}^n)}[ T_L \in
         \left(R_1 \cap R_2\cap R_3\right) ]\\
  \geq & \Pr_{T_L \sim_{i.i.d.}(\cE_{S_L}^n)}[ T_L \in R_3 ] - \Pr_{T_L
         \sim_{i.i.d.}(\cE_{S_L}^n)}[ T_L\in R_3 \mbox{ and } T_L\notin R_1 ] - \Pr_{T_L
         \sim_{i.i.d.}(\cE_{S_L}^n)}[ T_L\in R_3 \mbox{ and } T_L \notin R_2] \\
  \geq & 1/e - \beta - \beta >0
\end{align*}
Therefore, there exists a $T_L^* \in R_1 \in R_2$. Since $\cM$ is
perfectly generalizing, we have
that,
\begin{equation}\label{eq:equivesim}
\cM(T_L^*) \approx_{\eps, \delta} \Sim \quad \mbox{ and } \quad 
\cM(T_L^*) \approx_{\eps, \delta} \Sim'.
\end{equation}

This means with probability at least $1- 2\beta$, we also have
\[ \cM(T_L) \approx_{\eps, \delta} \Sim \approx_{\eps, \delta} \cM(T_L^*) \approx_{\eps, \delta} \Sim' \approx_{\eps, \delta} \cM(T_L'). \] 
By Lemma \ref{lem.structural}, with probability at least $1- 2\beta$,
\[
\cM'(S_L) = \cM(T_L)  \approx_{4\eps, 16\delta} \cM(T_L') = \cM'(S_L').
\]
Therefore, $\cM'$ is $(4\eps, 16\delta + 2\beta)$-differentially
private.
\end{proof}

\begin{theorem}\label{thm.pglearning}
Let $\cH$ be a hypothesis class with finite VC dimension $d$. Suppose that $\cH$ admits an agnostic learner $\cM\colon\cX_L^n
\rightarrow \cH$ that is $(\alpha, \gamma)$-accurate and $(\beta,
\eps, \delta)$-perfectly generalizing. Then algorithm $\cM'(\cM,
\cdot)$ defined as in~\Cref{alg.sgtodp} is $(4\eps, 16\delta +
2\beta)$-differentially private, and is also an $(O(\alpha),
O(\gamma))$-accurate agnostic learner for $\cH$.
\end{theorem}

We will rely on the following result on the uniform convergence
properties of any hypothesis class with finite VC dimension.

\begin{theorem}[see, e.g.,~\cite{shai} Theorem 6.8]\label{shai2}
Let $\cH$ be a hypothesis class of VC dimension $d < \infty$. Then
there are constants $C_1$ and $C_2$ such that the following holds:
\begin{enumerate}
\item Fix any $\alpha, \gamma > 0$. Let $S_L\sim_{i.i.d.}\cD_L^n$, then
  with probability at least $1-\gamma$, $|err(S_L, h) - err(h)| \leq
  \alpha$ for all $h\in\cH$, as long as
 $$n \geq C_1 \frac{d + \log(1/\gamma)}{\alpha^2}$$
\item Any agnostic learner that is $(\alpha, \gamma)$-accurate
  requires a sample of size  $$n\geq C_2 \frac{d +
    \log(1/\gamma)}{\alpha^2}$$
\end{enumerate}
\end{theorem}

\begin{proof}[Proof of~\Cref{thm.pglearning}]
  Let $S_L\sim_{i.i.d.}\cD_L^n$ be a random sample of size $n$.  By Part
  2 of~\Cref{shai2} and our assumption that $\cM$ is an $(\alpha,
  \gamma)$-accurate agnostic learner, we know that $n\geq C_2
  \frac{(d+\log(1/\gamma))}{\alpha^2}$. 
  By Part 1 of~\Cref{shai2}, we have with probability at least
  $1-\gamma$ over the random draws of $S_L$, for each $h\in \cH$,
\begin{equation}\label{morejunk}
  |err(S_L,  h) - err(h)| \leq O(\alpha).
\end{equation}
Let $\hat h = \cM'(\cM, S_L)$.  First, we can view $\cE_{S_L}$ as some
distribution over the labelled examples.  Since $\cM$ is an
$(\alpha,\gamma)$-accurate learner, we have with probability at least
$1-\gamma$,
\begin{equation}
 err(S_L, \hat h) \leq \min_{h\in \cH} err(S_L, h) + \alpha.
\label{junk}
\end{equation}
Let's condition on  guarantee of both~\Cref{morejunk,junk}.
Let $h^* = \arg\min_{h\in \cH} err(h)$. Then by
combining~\Cref{junk,morejunk}, we get
\[
err(\hat h) \leq err(S_L, \hat h) + O(\alpha) \leq 
err(S_L, h^*) + O(\alpha) \leq err(h^*) + O(\alpha)
\]
which recovers the stated utility guarantee.  By~\Cref{thm.sgtodp},
know that the mechanism $\cM'(\cM, \cdot)$ is also $(4\eps, 16\delta +
2\beta)$-differentially private.
\end{proof}

The result of~\Cref{thm.pglearning} implies that the existence of a
perfectly generalizing agnostic learner would imply the existence of a
differentially private one. Moreover, the lower bound results for
private learning would  apply to a perfectly generalizing learner as
well. In particular, based on the result of~\cite{BNSV15}, we can
show that there is no proper threshold learner that satisfies
perfect generalization when the domain size is infinite.

\begin{theorem}
Let $\alpha > 0$ be the accuracy parameter. For every $n\in \NN$, and
$\delta, \beta \leq 1/(10000n^2)$, any $(\beta, 1/8,
\delta)$-perfectly generalizing and $(\alpha, 1/32)$-accurate proper
agnostic learner for threshold function requires sample complexity $n
= \Omega\left(\log^*|\cX|/\alpha\right)$.
 \end{theorem}

\section{Relationship between Perfect Generalization and Other Generalization Notions}

In the previous sections we have studied the three generalization
notions as learnability constraints, and we know that any class that
learnable under perfect generalization is also learnable under
differential privacy, and any class learnable under differential
privacy is also learnable under robust generalization.  In this
section, we study these three notions from the algorithmic point of
view, and explore the relationships among algorithms that satisfy
perfect generalization, robust generalization and differential
privacy.  \Cref{s.pgandrg} shows that any perfectly generalizing
algorithms is also robustly generalizing, but there exist robustly
generalizing algorithms that are neither differentially private nor
perfectly generalizing for any reasonable parameters.
\Cref{s.pganddp} shows that all differentially private algorithms are
perfectly generalizing with some necessary loss in generalization
parameters, but there exist perfectly generalizing algorithms which
are not differentially private for any reasonable parameters.


\subsection{Separation between Perfect and Robust Generalization}\label{s.pgandrg}

In this section we show that perfect generalization is a stronger requirement than robust generalization.  Lemma \ref{lem.pgisrg} shows one direction of this, by showing that every perfectly generalizing mechanism also satisfies robust generalization with only a constant degradation in the generalization parameters.


\begin{lemma}\label{lem.pgisrg}
For any $\beta, \eps, \delta \in (0,1)$, suppose that a
mechanism $\cM\colon \cX_L^n\rightarrow \cR$ with arbitrary range $\cR$ is $(\beta, \eps, \delta)$-perfectly generalizing. Then $\cM$ is also $(\alpha, 2(\beta +\delta))$-robustly generalizing, where
\[
\alpha = \sqrt{\frac{2}{n} \ln\left(\frac{2(2\eps +
    1)}{\beta + \delta} \right)}.
\]
\end{lemma}

\begin{proof}
Let $\cA\colon\cR\rightarrow (\cX \rightarrow \{0, 1\})$ be
function that takes in the output of $\cM(S_L)$ and produces a hypothesis
$h \colon \cX \rightarrow \{0,1 \}$. Our goal is to show that $h$ will not overfit to the original sample
$S_L$.

By~\Cref{lem:post}, the composition of $\cA\circ \cM\colon \cX^n
\rightarrow (\cX_L \rightarrow \{0,1\})$ is also $(\beta, \eps, \delta)$-perfectly generalizing. This means there exists a simulator $\Sim \colon \Delta \cX \rightarrow \cR$
such that with high probability over a random sample $S_L$,
$\Sim(\cD)\approx_{\eps, \delta} (\cA\circ \cM)(S_L)$. Define the event $E =
\{S_L\in \cX^n \mid \left[\Sim(\cD) \approx_{\eps, \delta}
  (\cA\circ\cM)(S_L)\right]\}$.  By perfect generalization, $\Pr_{S_L\sim_{i.i.d.}\cD_L^n}[E]\geq
1 - \beta$.

Also by a Chernoff-Hoeffding bound, for any fixed $h\in \cH$ and any
$\alpha > 0$,
\begin{align*}
  \Pr_{S\sim_{i.i.d.}\cD^n}[|h(S) - h(\cD)| \geq \alpha] \leq
  2\exp\left( -2 \alpha^2 n\right).
\end{align*}

The following bounds the probability that the hypothesis $h$ output by $(\cA\circ \cM)(S_L)$ overfits on the sample $S_L$, where $\wedge$ denotes the logical AND.
\begin{align*}
\Pr_{S_L\sim_{i.i.d.}\cD_L^n}&[h \gets (\cA\circ\cM(S_L)) \wedge |h(S) - h(\cD)|\geq \alpha]\\
&=\sum_{S_L\in \cX_L^n} \Pr[S] \Pr[h \gets (\cA\circ\cM(S_L)) \wedge |h(S) - h(\cD)|\geq \alpha \; | \; S]\\
&\leq (1- \Pr[E]) + \sum_{S\in E}\Pr[S] \Pr[h \gets (\cA\circ\cM(S)) \wedge |h(S) - h(\cD)|\geq \alpha \; | \; S]\\
&\leq (1- \Pr[E]) + \sum_{S\in E}\Pr[S] \left(\Pr[h \gets \Sim(\cD) \wedge |h(S) - h(\cD)|\geq \alpha \; | \; S] \cdot \exp(\eps) + \delta\right)\\
&\leq (1- \Pr[E]) + \sum_{S\in \cX^n}\Pr[S] \left(\Pr[h \gets \Sim(\cD) \wedge |h(S) - h(\cD)|\geq \alpha \; | \; S] \cdot \exp(\eps) + \delta\right)\\
&= (1- \Pr[E]) + \delta  + \exp(\eps)\Pr_{S\sim_{i.i.d.}\cD^n}[h \gets \Sim(\cD) \wedge |h(S) - h(\cD)|\geq \alpha]\\
&\leq (1- \Pr[E]) + \delta  +2 \exp(\eps)\cdot\exp(- 2\alpha^2 n) \\
&\leq  \beta + \delta + 2 \exp(\eps)\cdot\exp(- 2\alpha^2 n) 
\end{align*}

Setting $\alpha = \sqrt{\frac{2}{n} \ln\left(\frac{2(2\eps +1)}{\beta + \delta} \right)}$ also gives $\exp(-2\alpha^2n) = \frac{\beta + \delta}{2(2\eps +1)}$.  Plugging this into the above equations, we see that the probability that $(\cA\circ \cM)(S_L)$ overfits to $S_L$ by more than our choice of $\alpha$ is at most
\[ \Pr_{S_L\sim_{i.i.d.}\cD^n}[h \gets (\cA\circ\cM(S_L)) \wedge |h(S) - h(\cD)|\geq \alpha] \leq \beta + \delta + 2\exp(\eps) \frac{\beta + \delta}{2(1 + 2\eps)}= 2(\beta + \delta). \]

Thus $\cM$ is $(\alpha, 2(\beta + \delta))$-robustly generalizing for our specified value of $\alpha$.
\end{proof}

Our next result, Lemma \ref{lem.rgnotpg}, shows that there exist robustly generalizing mechanisms that are neither differentially private nor perfectly generalizing, for any reasonable parameters.

\begin{lemma}\label{lem.rgnotpg}
For any $\gamma > 0$ and $n\in \NN$, there exists a mechanism
$\cM\colon\cX_L^n\rightarrow \{0,1\}$ that is $(\sqrt{\ln(2/\gamma)/2n},
\gamma)$-robustly generalizing, but is not $(\eps,
\delta)$-differentially private for any bounded $\eps$ and $\delta<1$,
and is not $(\beta, \eps', \delta')$-perfectly generalizing for any
$\beta < 1/2 - 1/\sqrt{n}$, bounded $\eps'$, and $\delta' < 1/2$.
\end{lemma}

\begin{proof}
Consider the domain $\cX = \{0, 1\}$, and the following deterministic mechanism
$\cM\colon\cX^n \rightarrow \{0,1\}$: given a sample $S$, output 1 if
more than $\lfloor n/2 \rfloor$ of the elements in $S$ is 1, and output 0 otherwise.
Note $\cM$ has a small output space, so by~\Cref{thm:description}, $\cM$ is $(\sqrt{\ln(2/\gamma)/2n}, \gamma)$-robustly generalizing for
any $\gamma > 0$.

Consider two neighboring samples $S_1$ and $S_2$ such that $S_1$ has $\lfloor n/2 \rfloor +1$
number of 1's, and $S_2$ has $\lfloor n/2 \rfloor$ number of 1's. Then $\Pr[\cM(S_1) =1] = 1$ and
$\Pr[\cM(S_2) = 1] =0$. Therefore, the mechanism is not $(\eps,
\delta)$-differentially private for any bounded $\eps$ and $\delta
<1$.

To show that $\cM$ is not perfectly generalizing, consider the
distribution $\cD$ that is uniform over $\cX = \{0,1\}$.  That is, $\Pr_{x\sim\cD}[x =
  1] = \Pr_{x\sim\cD}[x=0] = 1/2$. Suppose that $\cM$ is $(\beta, \eps',
\delta')$-perfectly generalizing with $\beta < 1/2 -
1/\sqrt{n}$. In particular, this implies that $\beta < 1/2 - \frac{{n
    \choose \lfloor n/2 \rfloor}}{2^n}$. Let $\Sim$ be the associated simulator, and
let $p = \Pr[\Sim(\cD) = 1]$.

Since each the events of $(\cM(S) = 0)$ and $(\cM(S) = 1)$
will occur with probability (over the random draws of $S$) greater than
$\beta$, then there exist samples $S_1$ and $S_2$ such that both $\cM(S_1), \cM(S_2) \approx_{\eps', \delta'} \Sim(\cD)$, and furthermore $\cM(S_1) = 1$ and $\cM(S_2) =0$ deterministically. This means, 
\[
p \leq \exp(\eps')\cdot \Pr[\cM(S_2) = 1] + \delta' =  \delta' \quad \mbox{ and, } \quad
(1 - p) \leq \exp(\eps')\cdot \Pr[\cM(S_1) = 0] + \delta' = \delta'.
\]
It follows from above that $\delta' \geq 1/2$. Thus, $\cM$ is not
$(\beta, \eps', \delta')$ for any $\beta < 1/2 - 1/\sqrt{n}$, bounded $\eps'$, and $\delta' < 1/2$.
\end{proof}

\subsection{Perfect Generalization and Differential Privacy}\label{s.pganddp}

We now focus on the relationship between differential privacy and
perfect generalization to show that perfect generalization is a strictly stronger definition in the sense that \emph{problems} that can be solved subject to perfect generalization can also be solved subject to differential privacy with little loss in the parameters, whereas in the reverse direction, parameters necessarily degrade.  Recall that we have already shown that any perfectly generalizing algorithm can be ``compiled'' into a differentially private algorithm with only a constant factor loss in parameters (Theorem \ref{thm.sgtodp}). We here note however that this compilation is necessary -- that perfectly generalizing algorithms are not necessarily themselves differentially private.  In the reverse direction, we show that every differentially private algorithm is strongly generalizing, with some necessary degradation in the generalization parameters.



We first give an example of a perfectly generalizing algorithm that does not satisfy differential privacy for any reasonable parameters.  The intuition behind this result is that perfect generalization requires an algorithm to behave similarly only on a $(1-\beta)$-fraction of samples, while differential privacy requires an algorithm to behave similarly on all neighboring samples.  The algorithm of Theorem \ref{thm.sgnotdp} exploits this difference to find a pair of unlikely neighboring samples which are treated very differently.

\begin{theorem}\label{thm.sgnotdp}
For any $\beta > 0$ and any $n \geq \log(1/\beta)$, there exists a
algorithm $\cM \colon \cX^n \rightarrow \cR$ which is $(\beta, 0, 0)$-perfectly generalizing but
is not $(\eps, \delta)$-differentially private for any $\eps < \infty$
and $\delta < 1$.
\end{theorem}
\begin{proof}
Consider the domain $\cX = \{0, 1\}$ and the following simple
algorithm $\cM$: given a sample $S = \{s_1, \ldots, s_n\}$ of size $n$,
it will output ``Strange'' if the sample $S$ satisfies:
\[
s_1 = s_2 = \ldots = s_{\lfloor n/2 \rfloor} = 1 \qquad \mbox{ and, } \qquad
s_{\lfloor n/2 \rfloor + 1} = s_{\lfloor n/2 \rfloor + 2} = \ldots = s_{n} = 0,
\]
 and output ``Normal'' otherwise. We first show that $\cM$ is
 $((1/2)^n, 0, 0)$-perfectly generalizing. Consider the following
 deterministic simulator $\Sim$ that simply outputs ``Normal'' no
 matter what the input distribution over the domain is.

Suppose that the distribution $\cD$ over the domain satisfies
$\Pr_{x\sim \cD}[x = 1] = p$ for some $p\in [0, 1]$. Note that the
probability (over the random draws of $S$) of outputting ``Strange'' is
\[
\Pr[\cM(S) = \text{``Strange''}] = p^{\lfloor n/2 \rfloor}(1 - p)^{\lceil n/2 \rceil} = \leq (1/2)^n.
\]
This means, with probability at least $1- (1/2)^n$ over the random
draws of $S$, $\cM$ will output ``Normal,'' and also
\[
\frac{\Pr[M(S) = \text{``Normal''}]}{\Pr[\Sim(\cD) = \text{``Normal''}]} = 1 \leq \exp(0).
\]
Therefore, $\cM$ is $((1/2)^n, 0, 0)$-perfectly generalizing.

Now consider the sample $T = \{t_1,\ldots , t_n\}$ such that
\[
t_1 = t_2 = \ldots = t_{\lfloor n/2 \rfloor} = 1 \qquad \mbox{ and, } \qquad
t_{\lfloor n/2 \rfloor + 1} = t_{\lfloor n/2 \rfloor + 2} = \ldots = t_{n} = 0.
\]
Let $T'$ be any neighboring sample of $T$ such that $|T\Delta
T'|=1$. We know that $\cM(T') = \text{``Normal''}$, so,
\[
\frac{\Pr[\cM(T') = \mbox{``Normal''}]}{\Pr[\cM(T) = \mbox{``Normal''}]} = \frac{1}{0} = \infty.
\]
Therefore, the algorithm $\cM$ is not $(\eps, \delta)$-differentially
private for any $\eps<\infty$ and $\delta < 1$.
\end{proof}


Now we show the other direction of the relationship between these two definitions: any differentially private algorithm is also perfectly generalizing.  We begin with Theorem \ref{thm.dptosg}, which proves that every $(\epsilon, 0)$-differentially private algorithm is also $(\beta, O(\sqrt{n \ln (1/\beta)} \eps), 0)$-perfectly generalizing.  We will later show that this dependence on $n$ and $\beta$ is tight.

\begin{theorem}\label{thm.dptosg}
Let $\cM\colon \cX^n \rightarrow \cR$ be an $(\eps, 0)$-differentially
private algorithm, where $\cR$ is an arbitrary finite range. Then $\cM$ is also
$(\beta, \sqrt{2n \ln(2|\cR| /\beta)} \eps, 0)$-perfectly generalizing.
\end{theorem}


\begin{proof}
Given an $(\eps, 0)$-differentially private algorithm $\cM$, consider the
following log-likelihood function $q\colon \cX^n \times \cR
\rightarrow \RR$, such that for any sample $S\in \cX^n$ and
outcome $r\in \cR$, we have
\[ q(S, r) \overset{\text{def}}{=} \log\left( \Pr[\cM(S) = r]\right). \]
Since $\cM$ is $(\eps, 0)$-differentially private, we know that for all neighboring $S, S' \in \cX^n$, the function $q$ satisfies,
\[ \max_{r \in \cR} \left| q(S, r) - q(S', r) \right| =\max_{Œr\in \cR} \left|\ln\left( \frac{\Pr[\cM(S) = r]}{\Pr[\cM(S') = r]} \right) \right| \leq \eps. \]

For any distribution $\cD \in \Delta \cX$, the sample $S = (s_1, \ldots, s_n) \sim_{i.i.d.} \cD^n$ is now a random variable, rather than a fixed input.  By an application of McDiarmid's inequality \rc{add to the appendix} to the variables $s_1, \ldots s_n$, we have that for any $r\in \cR$,
\begin{equation}\label{eq.mcdiarmid} \Pr_{S\sim_{i.i.d.} \cD^n} \left[
    \left| q(S, r) - \Ex{{S' \sim_{i.i.d.}  \cD^n}}{q(S', r)}\right|
    \geq t \right] \leq 2 \exp\left(\frac{-2t^2}{n\eps^2} \right).
\end{equation}
Instantiating Equation~\eqref{eq.mcdiarmid} with $t = \eps \sqrt{(n/2)\ln(2|\cR|/\beta)}$ and taking a union bound, we have that with probability at least $1-\beta$, it holds for all $r \in \cR$ that,
\begin{equation}\label{eq.expmech} \left| q(S, r) - \Ex{{S'
        \sim_{i.i.d.}  \cD^n}}{q(S', r)} \right| \leq \eps
  \sqrt{(n/2)\ln(2|\cR|/\beta)}.
\end{equation}

Define the simulator $\Sim(\cD)$ for algorithm $\cM$ on distribution
$\cD$ as follow for all $r \in \cR$, output the $r$ with probability
proportional to
$\exp\left(\Ex{S \sim_{i.i.d.} \cD^n}{q(S, r)}\right)$. Let

$$
Z = \frac{\sum_{r\in \cR}\exp\left(\Ex{S' \sim_{i.i.d.}  \cD^n}{q(S',
      r)}\right)}{\sum_{r\in \cR}\exp\left( {q(S, r)}\right)}
$$
be the ratio between the normalization
factors, and by~\Cref{eq.expmech},
$$\exp\left(- \eps \sqrt{(n/2)\ln(2|\cR|/\beta)} \right) \leq Z \leq \exp\left( \eps \sqrt{(n/2)\ln(2|\cR|/\beta)} \right)$$

We condition on the bound in~\Cref{eq.expmech} for the remainder of
the proof, which holds except with probability $\beta$.  For any
$r \in \cR$,
\begin{align*}
\frac{\Pr[ \cM(S) = r] }{\Pr [ \Sim(\cD) = r]} &= \frac{\exp\left( q(S,r) \right)}{ \exp\left( \E_{S' \sim_{i.i.d.} \cD^n} [q(S',r)] \right) /Z } \\
&= \exp \left( q(S,r) - \E_{S' \sim_{i.i.d.} \cD^n} [q(S',r)] \right) \cdot Z \\
&\leq \exp \left( \eps \sqrt{2n\ln(2|\cR|/\beta)} \right),
\end{align*}
where the last inequality is due to Equation~\eqref{eq.expmech}.

For any $\cO \subseteq \cR$ and for $\eps' = \eps \sqrt{2n\ln(2|\cR|/\beta)}$,
\begin{align*}
\Pr [ \cM(S) \in \cO] &= \sum_{r \in \cO} \Pr [ \cM(S) = r] \\
&\leq \sum_{r \in \cO} e^{\eps'} \Pr [ \Sim(\cD) = r] \\
&= e^{\eps'} \Pr [ \Sim(\cD) \in \cO].
\end{align*}
Similarly, we could also show
\[ \frac{ \Pr[ \Sim(\cD) \in \cO]}{ \Pr[ \cM(S) \in \cO]} \leq \exp\left(\eps \sqrt{2n\ln(2|\cR|/\beta)} \right). \]

Thus for any distribution $\cD \in \Delta \cX$, with probability at
least $1-\beta$ over the choice of $S \sim_{i.i.d.} \cD^n$, we have
that $\cM(S) \approx_{\eps', 0} \Sim(\cD)$, for
$\eps' = \eps \sqrt{2n\ln(2|\cR|/\beta)}$, so $\cM$ is
$(\beta, \eps \sqrt{2n\ln(2|\cR|/\beta)}, 0)$-perfectly
generalizing.
\end{proof}

The following result proves that the degradation of $\eps$ in Theorem \ref{thm.dptosg} is necessary, and the dependence on $n$ and $\beta$ is asymptotically tight.


\begin{theorem}\label{thm.dptosgtight}
For any $\eps > 0$, $\beta \in (0, 1)$ and $n\in \NN$, there exists a
algorithm $\cM \colon \cX^n \rightarrow \cR$ that is $(\eps, 0)$-differentially private, but not
$(\beta, \eps', 0)$-perfectly generalizing for any $\eps' = o(\eps\sqrt{n \ln(1/\beta)})$.
\end{theorem}

\begin{proof}
  Consider the domain $\cX = \{0, 1\}$ and the distribution $\cD$ over
  $\cX$ such that
  $\Pr_{x\sim\cD}[x = 1] = \Pr_{x\sim\cD}[x = 0] = 1/2$. Consider
  following algorithm $\cM \colon \cX^n\rightarrow \{0, 1\}$.  Given a
  sample $S = \{s_1, \ldots, s_n\}\in \cX^n$, $\cM$ will do the
  following:
\begin{enumerate}
\item first compute the sample average $\overline s = \frac{1}{n}\sum_{i=1}^n s_i$;
\item then compute a noisy estimate $\hat s = \overline s +
  \Lap(\frac{1}{n\eps})$ by adding Laplace noise with parameter $1/ n \eps$;
\item if $\hat s \leq 1/2$, output 0; otherwise, output 1.
\end{enumerate}
In words, the algorithm tries to identify the majority in the
sample based on the noisy estimate $\hat s$. Note that the average
value $\overline s$ is a $(1/n)$-sensitive statistic --- that is,
changing a single sample point $s_i$ in $S$ will change the value of $\overline
s$ by at most $1/n$. Also observe that $\cM$ is the Laplace mechanism of~\cite{DMNS06}
composed with a (data independent) post-processing step, so we know $\cM$ is $(\eps, 0)$-differentially
private.

Now suppose that $\cM$ is $(\beta, \eps', 0, n)$-strongly generalizing
for some $\beta \in (0, 1)$. Using a standard tail bound for the Binomial
distribution, we know that for any $S\sim_{i.i.d.}\cD^n$ and $k \leq
1/8$, the sample average $\overline s$ satisfies\sw{cite! }
\[
\Pr[\overline s \leq n/2 - k] = \Pr[\overline s \geq n/2 + k] \geq
\frac{1}{15}\exp\left( -{16 n k^2}\right).
\]
In other words, for any $\gamma \in (0, 1)$, we have both
$\Pr[\overline s \leq 1/2 - K] \geq \gamma$ and $\Pr[\overline s \geq
  1/2+K] \geq \gamma$, where $K =
\frac{\sqrt{\ln(1/(15\gamma))}}{4\sqrt{n}}$. For the remainder of the
proof, we will set $\gamma = 2\sqrt \beta$.

Let $S_1, S_2\sim_{i.i.d.}\cD^n$ be two random samples with sample
averages $\overline s_1$ and $\overline s_2$. By Corollary
\ref{cor:twin}, we know that
$\Pr[\cM(S_1)\not\approx_{2\eps', 0} \cM(S_2)] \leq
2\beta$.  
Since $\gamma^2 > 2\beta$, it follows that with strictly positive
probability over the random draws over $S_1$ and $S_2$, all of the
events that $\overline s_1 \leq n/2 - K$,
$\overline s_2 \geq n/2 + K$, and
$\cM(S_1)\approx_{2\eps', 0} \cM(S_2)$ occur simultaneously.  For the
remainder of the proof, we condition on samples $S_1$ and $S_2$
satisfying these conditions, which will happen with probability
greater than $2\beta$.

If we apply our algorithm $\cM$ to both samples, we will first obtain
noisy estimates $\hat s_1$ and $\hat s_2$, and by the property of the
Laplace distribution, we know for any $t>0$
\[
\Pr\left[|\hat s_1 - \overline s_1| < K\right] = 1 - \exp(-K n\eps) \qquad \mbox{and}\qquad
\Pr\left[|\hat s_2 - \overline s_2| < K\right] = 1 - \exp(- K n\eps)
\]
Note that the event $|\hat s_1 - \overline s_1| < K$ implies that
$M(S_1) = 0$, and the event $|\hat s_2 - \overline s_2| < K$ implies
that $M(S_2) = 1$. The condition of $M(S_1) \approx_{2\eps',
  0} M(S_2)$ implies that
\[
\exp(2\eps') \geq \frac{\Pr[M(S_1) = 0]}{\Pr[M(S_2) = 0]} =\frac{\Pr[M(S_1) =
    0]}{1 - \Pr[M(S_2) = 1]} \geq \frac{1 - \exp(- K n\eps)}{\exp(- K n\eps)} = \exp(K n\eps) - 1
\]

It follows that we must have
\[
\eps' \geq \frac{1}{2}(K n \eps - 1) = \Omega\left(\eps \sqrt{n \ln(1/\beta)} \right),
\]
which recovers the stated bound.
\end{proof}

Theorems \ref{thm.dptosg} and Theorem \ref{thm.dptosgtight} only show
a relationship between $(\eps, 0)$-differential privacy and strong
generalization.  To show such a relationship when $\delta > 0$, we
appeal to \emph{group privacy}, first studied by \cite{DKM+06}, which
says that if $\cM$ is $(\eps, \delta)$-differentially private and two
samples $S, S'$ differ on $k$ entries, then
$\cM(S) \approx_{k \eps, k e^{(k-1)\eps} \delta} \cM(S')$.  Using
simulator $\Sim_{\cD} = \cM(S^*)$ for any fixed sample
$S^* \sim_{i.i.d.} \cD^n$ and by the fact that any sample $S$ can
differ from $S^*$ in an most $n$ samples, we see that $\cM$ is
$(0, n \eps, n e^{(n-1) \eps} \delta)$-perfectly generalizing.

Unfortunately, this blowup in parameters is generally unacceptable for
most tasks. We suspect that the necessary blowup in the $\eps$
parameter is closer to $\Theta\left(\sqrt{n \ln(1/\beta)}\right)$ as
with $(\eps, 0)$-differential privacy, but leave a formal proof as an
open question for future work.

On the positive side, most known $(\eps, \delta)$-differentially
private algorithms are designed by composing several
$(\eps', 0)$-differentially private algorithms, where the $\delta>0$
is an artifact of the composition (see, e.g., Theorem 3.20 of
\cite{DR14} for more details).  Since perfect generalization enjoys
adaptive composition (as shown in~\cite{BF16}), we could also obtain
$(\beta, \eps, \delta)$-perfectly generalizing algorithms by composing
a collection of $(\beta, \eps, 0)$-perfectly generalizing algorithms
together. This will give better generalization parameters than a
direct reduction via group privacy.

\subsection*{Acknowledgements}
We thank Adam Smith and Raef Bassily for helpful comments about
adaptive composition of perfectly generalizing mechanisms, and for pointing
out an error in an earlier version of this paper. We thank Shay Moran
for telling us about variable-length compression schemes and sharing
with us his manuscript \cite{DMY16}. We thank our anonymous reviewers
for numerous helpful comments.
\bibliographystyle{alpha}
\bibliography{main.bbl}

\appendix

\section{Missing Proofs in~\Cref{sec:prelim}}\label{s.prelimproofs}



\begin{proof}[Proof of Lemma \ref{lem.structural}]

In the following, we will use $(a \wedge b)$ to denote $\min\{a, b\}$.
For all $\cO\subseteq \cR$,
\begin{eqnarray*}
\Pr_{y\sim\cJ_1}\left[ y \in \cO \right] & \leq &
(\exp(\eps) \Pr_{y\sim\cJ_2}\left[y \in \cO\right] + \delta) \wedge 1 \\
&\leq &(\exp(\eps)\Pr_{y\sim\cJ_2}\left[y \in \cO\right]) \wedge 1 + \delta\\
& \leq & \exp(\eps) \left(\exp(\eps') \Pr_{y\sim\cJ_3}\left[y \in \cO\right] +\delta'\right)+ \delta \\
& = & \exp(\eps+\eps') \Pr_{y\sim\cJ_3}\left[y \in \cO\right] + 2\delta' + \delta.
\end{eqnarray*}
A similar argument gives $\Pr_{y\sim\cJ_3}\left[ y \in \cO \right] \leq \exp(\eps+\eps') \Pr_{y\sim\cJ_1}\left[y \in \cO\right] + 2\delta + \delta'$.
\end{proof}


\begin{proof}[Proof of Corollary \ref{cor:twin}]
By a union bound, with probability $1-2\beta$ over the draws of $T_1, T_2\sim_{i.i.d.}\cC^n$,
\[ \cM(T_1) \approx_{\eps, \delta} \Sim_{\cC}  \approx_{\eps, \delta} \cM(T_2). \]
The result then follows from Lemma \ref{lem.structural}.
\end{proof}


\begin{proof}[Proof of Lemma \ref{lem:post}]
The result for robustly generalizing mechanisms follows immediately from the definition:
Assume by way of contradiction that there exists an $(\alpha, \beta)$-robustly generalizing mechanism $\cM \colon \cY^n \rightarrow \cR$ and a post-processing procedure $\cA\colon \cR \rightarrow \cR'$ such that $\cA\circ\cM$ is not $(\alpha, \beta)$-robustly generalizing.  Then there exists an adversary $\cA'$ such that $\cA' ( \cA (M(T)))$ outputs a hypothesis $h$ that violates the robust generalization condition.  However,  this would imply that $\cA' \circ \cA$ is an adversary that violates the robust generalization condition, contradicting
 the assumption that $\cM$ is $(\alpha, \beta)$-robustly generalizing.

Let $\cM \colon \cY^n \rightarrow \cR$ be $(\beta, \eps, \delta)$-perfectly generalizing, and let $\cA\colon \cR \rightarrow \cR'$ be a post-processing procedure.  Fix any distribution $\cC$, and let $\Sim_{\cC}$ denote the simulator such that $\cM(T) \approx_{\eps, \delta} \Sim_{\cC}$ with probability $1-\beta$ when $T\sim_{i.i.d.}\cC^n$. We will show that with probability at least $1-\beta$ over the sample $T\sim_{i.i.d.}\cC^n$,
\[ \cA(\cM(T)) \approx_{\eps, \delta} \cA(\Sim_{\cC}). \]

First note that any randomized mapping can be decomposed into a convex combination of deterministic mappings.  Let
\[ \cA = \sum_{i=1} \gamma_i \cA_i \quad \mbox{ s.t. } \quad \sum_{i=1} \gamma_i = 1 \mbox{ and } 0 < \gamma_i \leq 1 \; \forall i, \]
where each $\cA_i \colon \cR \rightarrow \cR'$ is deterministic.  For the remainder of the proof, we will assume that $\cM(T) \approx_{\eps, \delta} \Sim_{\cC}$, which will be the case with probability $1-\beta$.

Fix an arbitrary $\cO' \subseteq \cR'$ and define $\cO_i = \{ r \in \cR \; | \; \cA_i(r) \in \cO' \}$ for $i \in [k]$.
\begin{align*}
\Pr[ \cA(\cM(T)) \in \cO'] &= \sum_{i=1} \gamma_i \Pr[ \cA_i(\cM(T)) \in \cO'] \\
&= \sum_{i=1} \gamma_i \Pr[ \cM(T) \in \cO_i] \\
&\leq \sum_{i=1} \gamma_i \left( e^{\eps} \Pr[ \Sim_{\cC} \in \cO_i] + \delta \right) \\
&= \sum_{i=1} \gamma_i \left( e^{\eps} \Pr[ \cA_i(\Sim_{\cC}) \in \cO'] + \delta \right) \\
&= e^{\eps} \Pr[ \cA(\Sim_{\cC}) \in \cO'] + \delta.
\end{align*}
A symmetric argument shows that
\[ \Pr[ \cA(\Sim_{\cC}) \in \cO'] \leq e^{\eps} \Pr[ \cA(\cM(T)) \in \cO'] + \delta. \]

Thus with probability at least $1-\beta$, $\cA(\cM(T)) \approx_{\eps, \delta} \cA(\Sim_{\cC})$.  The mapping $\cA(\Sim_{\cC}) \colon \cY^n \rightarrow \cR'$ is simply a new simulator, so $\cA \circ \cM$ is $(\beta, \eps, \delta)$-perfectly generalizing.\end{proof}


\begin{proof}[Proof of \Cref{thm.basic}]
Fix any distribution $\cC$, and for all $i \in [k]$ let $\Sim_i(\cC)$ denote the simulator such that $\cM_i(T) \approx_{\eps, \delta} \Sim_i(\cC)$ with probability $1-\beta_i$ when $T\sim_{i.i.d.}\cC^n$.  Define $\Sim_{[k]}(\cC) = \left( \Sim_1(\cC), \ldots, \Sim_k(\cC) \right)$.  For the remainder of the proof, we will assume that $\cM_i(T) \approx_{\eps, \delta} \Sim_i(\cC)$ for all $i \in [k]$, which will be the case with probability at least $1 - \sum_{i=1}^k \beta_i$ over the choice of the sample.

Fix any $(r_1, \ldots, r_k) \in \cR_1 \times \cdots \times \cR_k$:
\begin{align*}
\Pr[\cM_{[k]}(T) = (r_1, \ldots, r_k)] &= \prod_{i=1}^k \Pr[ \cM_i(T) = r_i] \\
&\leq \prod_{i=1}^k e^{\eps_i} \Pr[ \Sim_i(\cC) = r_i] \\
&= e^{\sum_{i=1}^k \eps_i} \Pr[ \Sim_{[k]}(\cC) = (r_1, \ldots, r_k)] 
\end{align*}
For any $\cO \subseteq \cR_1 \times \cdots \times \cR_k$,
\begin{align*}
 \Pr[\cM_{[k]}(T) \in \cO] &= \int_{o \in \cO} \Pr[\cM_{[k]}(T) =o] do\\
& \leq \int_{o \in \cO} e^{\sum_{i=1}^k \eps_i} \Pr[\Sim_{[k]}(\cC) =o] do = e^{\sum_{i=1}^k \eps_i} \Pr[\Sim_{[k]}(\cC) \in \cO].
\end{align*}
A symmetric argument would show that $\Pr[\Sim_{[k]}(\cC) \in \cO] \leq e^{\sum_{i=1}^k \eps_i} \Pr[\cM_{[k]}(T) \in \cO]$.

The mapping $\Sim_{[k]}(\cC)$ serves as a simulator for $\cM_{[k]}(T)$, so $\cM_{[k]}$ is $(\sum_{i=1}^k \beta_i, \sum_{i=1}^k \eps_i, 0)$-perfectly generalizing.\end{proof}

\end{document}